\documentclass[11pt,singlecolumn,prx,longbibliography,nofootinbib, superscriptaddress]{revtex4-2}

\usepackage{config}
\usepackage{comment}
\usepackage{times}
\usepackage{tikz}

\usepackage{titlesec}
\titleformat{\paragraph}[runin]
  {\normalfont\bfseries}{\theparagraph.}{0.5em}{}[\;\;]
\titlespacing*{\paragraph}
  {0pt}{1.5ex plus 0.5ex minus 0.2ex}{0.5em}

\begin{document}

\title{Distinctness threshold for pseudorandom unitaries}

\author{Asad Raza}
\email{asad.raza@fu-berlin.de}
\affiliation{Dahlem Center for Complex Quantum Systems, Freie Universit\"at Berlin, 14195 Berlin, Germany}

\author{Jens Eisert}

\affiliation{Dahlem Center for Complex Quantum Systems, Freie Universit\"at Berlin, 14195 Berlin, Germany}

\affiliation{Helmholtz-Zentrum Berlin f{\"u}r Materialien und Energie, 14109 Berlin, Germany}

\author{Bill Fefferman}

\affiliation{Department of Computer Science, The University of Chicago}

\date{\today}
\begin{abstract}
Pseudorandomness is increasingly recognized as a key property of ensembles in quantum information theory, statistical mechanics, and quantum many-body physics. Yet it appears in two conceptually different forms: statistical pseudorandomness, embodied by unitary designs, and 
computational pseudorandomness captured by pseudorandom unitaries (PRUs). The relationship between these two forms of pseudorandomness remains surprisingly poorly understood. Existing PRU constructions reveal this interplay where a statistically randomizing ingredient—a unitary design—is combined with classical cryptographic primitives to produce computational pseudorandomness. 
We show that statistical pseudorandomness is not necessary for computationally pseudorandom unitaries. We do this by 
replacing the unitary $2$-design layer in the existing constructions with ensembles that are not even state $1$-designs, yet are sufficiently {\em distinct}, a property we identify to be necessary for any PRU. 
This yields new non-adaptively secure PRU ensembles whose computational pseudorandomness is obtained without an underlying statistically pseudorandom quantum ensemble, such as a $2$-design.
We characterize distinctness via an entangled analogue of anticoncentration and use it to show that distinctness already captures constraints on coherence and imaginarity of PRUs, while identifying broad classes of inputs for which the latter obstruction disappears, enabling real-valued PRUs even for certain (maximally) entangled states. As an application, we use lack of  distinctness to constrain the conjectured pseudorandomness of the random phase-Hadamard ensemble to form a PRU.
\end{abstract}

\maketitle 

\vspace{-0.3cm}
\tableofcontents

\section{Introduction}

Recent years have witnessed a remarkable convergence between two research directions that, at first sight, appear largely unrelated. On the one hand, quantum many-body physics has sought to understand the emergence of complexity in isolated quantum systems through phenomena such as thermalization, scrambling, quantum chaos, and information propagation \cite{1408.5148,hayden_black_2007,cotler_fluctuations_2022,christian_review}. On the other hand, theoretical computer science---and cryptography in particular---has developed computational notions of randomness, culminating in pseudorandom objects that efficiently emulate ideal random behaviour. Remarkably, these two perspectives have come to share many of the same mathematical structures. Random quantum circuits simultaneously model chaotic quantum dynamics, approximate Haar-random evolution through unitary designs \cite{Designs}, and are widely believed to realize cryptographic primitives such as pseudorandom quantum states and pseudorandom unitaries \cite{bostanci_efficient_2024}. This rapidly developing interplay has become one of the most exciting interfaces between quantum information, complexity theory, and many-body physics. One of the central outcomes of this convergence has been the theory of \emph{pseudorandom unitaries} (PRUs), first conjectured by Ji, Liu, and Song (JLS)~\cite{ji_pseudorandom_2018}. Informally, a PRU is an efficiently implementable ensemble of unitaries that is computationally indistinguishable from a Haar-random unitary to every efficient observer. Since their introduction, increasingly powerful constructions have been developed \cite{metger_simple_2024,schuster_strong_2025,lu_parallel_2025}, providing security under progressively stronger models of adversarial access. Despite this rapid progress, the existing constructions rely on rather different proof techniques. It has therefore remained difficult to determine which of their ingredients are fundamentally required for pseudorandomness and which merely provide convenient tools for establishing security. A common structure nevertheless appears across several existing proofs. Before computational pseudorandomness enters the argument, one first shows that the state obtained from $t$ parallel queries has almost all of its weight on the so-called \emph{distinct subspace}. This subspace is spanned by $n$-qubit computational basis states whose labels are pairwise distinct across the $t$ query registers. More precisely, the corresponding set of labels is \begin{equation} [N]_{\mathrm{dist}}^t := \left\{ (x_1,\ldots,x_t)\in[N]^t \,:\, x_i\neq x_j \text{ for all } i\neq j \right\}, \end{equation}
where  $N\coloneqq 2^n$. Operationally, concentration on this subspace means that measuring the query registers in the computational basis produces collisions only with small probability. In the celebrated $PFC$ family of constructions~\cite{metger_simple_2024,schuster_strong_2025,foxman_quantum_2026}, this property is obtained by applying a random Clifford and exploiting the fact that Clifford circuits form unitary $2$-designs. In constructions based on Kac's walk~\cite{lu_parallel_2025}, an analogous role is played by state $2$-designs. The distinct subspace has consequently been regarded primarily as a technical device in existing security proofs, generated by a much stronger form of statistical pseudorandomness. In this work, we show that distinctness is neither merely a proof artifact nor simply another manifestation of design-level randomness. Instead, it is the fundamental intermediate property connecting statistical randomization with computational pseudorandomness. 

 Throughout this work we will only focus on non-adaptive and forward query secure PRUs. We first prove that every pseudorandom unitary must place all but negligible weight on the distinct subspace for every efficiently preparable input state. Otherwise, collisions provide a direct efficient distinguisher from Haar random evolution. Distinctness is therefore a necessary threshold that every PRU must cross. Having established its necessity, we ask how much randomness is actually required to cross this threshold. Surprisingly, neither unitary nor state designs are necessary. We construct an ensemble consisting of a complex-valued random diagonal phase unitary, $F_{\C}$, followed by the $n$-qubit Hadamard transform, $H$, that achieves the same $O(t^2/2^n)$ distinctness bound as the unitary $2$-design, $C$, used in the $PFC$ construction \cite[Lemma 3.2]{metger_simple_2024}, while failing  to be a state $1$-design. Here $P$ and $F$ are random permutation and real-valued random binary phase operator respectively.
 
We further show that a single layer of independent single-qubit $2$-designs---for example, random single-qubit Clifford gates---is negligibly distinct for polynomially many queries. Within the $PFC$ construction, this single layer can replace the depth-$\log(n)$ global unitary $2$-design layer while still yielding a PRU. Designs therefore provide far more statistical randomness than is required at this stage of the construction. To give distinctness an operational characterization, we introduce \emph{entangled anticoncentration}. Standard anticoncentration \cite{dalzell_random_2022,hangleiter_anticoncentration_2018} requires the output probability mass not to concentrate on a small number of measurement outcomes and is ordinarily formulated for fixed product-state inputs. Entangled anticoncentration instead requires suppression of collisions for arbitrary bipartite inputs, including inputs entangled across the two query registers. In the computationally bounded regime of polynomially many queries and negligible error, we show that entangled anticoncentration is equivalent to distinctness. Distinctness thereby lies strictly below state $2$-designs while strengthening standard anticoncentration. In particular, our construction that is distinct but not even a state $1$-design shows that the recently established equivalence between anticoncentration and state $2$-designs on locally invariant architectures, such as brickwork Haar-random circuits \cite{heinrich_anti-concentration_2025}, does not extend to general architectures without local invariance. Once isolated as a necessary intermediate property, distinctness can be used in two complementary directions. Constructively, it allows the statistically randomizing layer in existing PRU constructions to be substantially weakened. Obstructively, failure of distinctness provides a simple test for ruling out candidate PRU constructions. The equivalence with entangled anticoncentration suggests a particularly direct distinguisher: prepare a Bell state across two query registers, apply the unknown unitary in parallel, and measure again in the Bell basis. We apply this test to the conjectured PRU construction of Ji, Liu, and Song~\cite{ji_pseudorandom_2018}, which consists of alternating random phase and Hadamard unitaries. While the distinguisher does not break the original conjectured construction, it rules out a broad weaker version in which the codomains of the random functions defining the phase unitaries are smaller than their domains by a superpolynomial factor. The same operational perspective also clarifies the quantum resources required for pseudorandomness. Using maximally entangled Bell states as witness inputs, we show that distinctness alone already enforces strong lower bounds on imaginarity and coherence. 

We thereby recover the resource lower bounds for PRUs established by Haug \emph{et al.}~\cite{haug_pseudorandom_2024} at the much weaker level of distinctness. In particular, an ensemble that is distinct on arbitrary inputs must be highly complex-valued; fully real unitary ensembles cannot satisfy this unrestricted requirement. This obstruction, however, turns out to be strongly input dependent. The impossibility of real-valued distinctness is witnessed specifically by inputs having large overlap with the canonical maximally entangled Bell state, rather than by entanglement itself. We formalize this observation through a Bell-overlap condition and construct real-valued distinct ensembles on inputs with sufficiently small Bell overlap. This condition includes all states with positive partial transpose, as well as broad classes of states with negative partial transpose, including certain maximally entangled states. As a consequence, we obtain concrete real-valued PRUs that are secure on such restricted classes of inputs, resolving an open question of Brakerski and Magrafta \cite{brakerski_real-valued_2024}. Our results also go beyond the PPT (positive patial transpose) condition previously imposed in Ref.~\cite{grevink_will_2025} for the equivalence between unitary and orthogonal designs. 

Taken together, these results identify distinctness as a threshold for separating the statisitical quantum pseudorandomness from quantum computational pseudorandomness in the context of PRUs. It is necessary for every PRU in the parallel forward-query setting and sufficient as the statistical randomization layer in the $PFC$ framework, yet dramatically weaker than a unitary or state design. This perspective turns distinctness into both a construction principle and a no-go test for quantum pseudorandomness. In \Cref{subsubsec:dist-prop-necc-charac}, we formally define distinctness, establish its necessity for PRUs, and characterize it through entangled anticoncentration. We then detail its applications in \Cref{subsubsec:apps-of-distinctness}.

\subsection{Main results}

 We now summarize our main contributions. We first identify distinctness as a necessary property for PRUs, and provide an operational characterization of distinctness in terms of an entangled version of anticoncentration. We then show that distinctness can arise without unitary or state designs and explore its consequences for PRU constructions, their resource requirements, and real-valued pseudorandomness.

\subsubsection{Distinctness as a property}
\label{subsubsec:dist-prop-necc-charac}    

\paragraph{Necessity of distinctness for PRUs and its operational characterization.} 

\begin{definition}[$\delta$-distinctness]\label{def:delta-distinctness}
    We say that an $n$-qubit unitary ensemble $\calE$ is $\delta$-distinct if for any $2 \leq t\leq N$,
    \begin{equation}
        \tr\left[\Pi^{\dist}\E_{U \leftarrow \mathcal{E}}\left[U^{\otimes t}\rho_{\mathsf{X}_1, \ldots , \mathsf{X}_t} U^{\otimes t, \dagger}\right]\right] \geq 1 - \delta,
     \end{equation}
    for all $nt$-qubit 
    quantum states $\rho_{\X_1, \ldots, \X_t}$, where $
\Pi^{\dist}
:=
\sum_{(x_1,\ldots,x_t)\in [N]^t_{\dist}}
|x_1,\ldots,x_t\rangle\langle x_1,\ldots,x_t|$.
\end{definition}

Distinctness can be understood in terms of the familiar anticoncentration property, as we will see in \Cref{thm:equiv-of-distinctness-and-EAC}. Recall that standard anticoncentration of an ensemble $\cal E$ requires that a uniformly randomly drawn unitary $U$ from $\cal E$ has, on average, small number of collisions when $U^{\ot 2}\ket{0,0}$ is being  measured in the computational basis. We define {\em entangled anticoncentration} (EAC) by enforcing the same few collision requirement, but for arbitrary bipartite input states, crucially the maximally entangled ones. 

\begin{definition}[$\delta$-(entangled anticoncentration) ($\delta$-EAC)]\label{def:delta-EAC}
 We say that an $n$-qubit unitary ensemble $\calE$ is $\delta$-(entangled anticoncentrated) if for any bipartite state $\omega_{\X_i, \X_j}$,
 \begin{equation}
    \tr\left[\Pi^{\eq}_{\X_i, \X_j} \E_{U\leftarrow \calE}\left[U^{\ot 2} \omega_{\X_i, \X_j} U^{\ot 2, \dagger}\right]\right] \leq \delta,
 \end{equation}
 where $
{\Pi}^{\eq}_{\X_i, \X_j} = \sum_{x\in [N]} \ketbra{x}{x}_{\X_i} \otimes \ketbra{x}{x}_{\X_j}$.
\end{definition} 

\begin{remark}
While $\delta$-distinctness in \Cref{def:delta-distinctness} (resp. $\delta$-EAC in \Cref{def:delta-EAC}) is defined for {\em arbitrary} $nt$-qubit (resp. $2n$-qubit) quantum state inputs, we sometimes require $\delta$-distinctness (resp. $\delta$-EAC) on a restricted set, $\cal S$, of quantum states. In that case, we will say that an ensemble $\calE$ is $\delta$-distinct (resp. $\delta$-EAC) on $\cal S$. Note that in the special case when $\cal S$ is the set of product states (typically just $\ketbra{0}{0}^{\ot 2}$), we recover the standard anticoncentration (for appropriately chosen $\delta>0$) as commonly used in the literature \cite{dalzell_random_2022, hangleiter_anticoncentration_2018}.
\end{remark}

Our first main result is to show the necessity of 
distinctness for PRUs, as opposed to being a mere proof artifact recurrently appearing in PRU constructions \cite{metger_simple_2024,schuster_strong_2025, lu_parallel_2025}.

\begin{namedtheorem}[\Cref{thm:PRU-must-be-negl-dist}]
    Any $n$-qubit pseudorandom unitary (PRU) ensemble $\calE$ must be $\negl(n)$-distinct on all efficiently preparable input states.
\end{namedtheorem}
Since this is established in the forward-only, non-adaptive query model, the necessary condition also applies to PRUs satisfying stronger notions of security.

We now formalize the relation between distinctness and entangled anticoncentration.
We show that, so long as the number of query registers $t$ are at most $\poly(n)$, which is the case for PRUs, $\delta$-EAC and $\delta$-distinctness are equivalent. We prove this equivalence in two steps. First, we show that $\delta$-distinctness implies $\delta$-EAC (\Cref{lemma:distinctness-implies-EAC}), for any $\delta$. It is then straightforward to show a reverse implication, but with a multiplicative factor of $t^2$. That is, $\delta$-EAC implies $(\delta\cdot t^2)$-distinctness (\Cref{lemma:EAC-implies-distinctness}).

\begin{theorem}[Equivalence of $\delta$-distinctness and $\delta$-EAC in the polynomial regime]\label{thm:equiv-of-distinctness-and-EAC}
    Let $\calE$ be some n-qubit unitary ensemble. Then in the computationally bounded regime, where $t=\poly(n)$ and any $\delta = \negl(n)$, $\calE$ is $\delta$-distinct iff $\calE$ is $\delta$-EAC.
\end{theorem}

\paragraph{Quantum resources for distinctness.}
The relationship between $\negl(n)$-distinctness and $\negl(n)$-EAC (\Cref{thm:equiv-of-distinctness-and-EAC}), beyond being conceptually useful, turns out to be crucial in proving quantum resource lower bounds for distinctness, and thus, for PRUs. We use this equivalence to strengthen the existing results \cite{haug_pseudorandom_2024} on quantum resource requirements for imaginarity and coherence for PRUs. In particular, we show that these resources are already implied by $\delta$-distinctness (and hence $\delta$-EAC) of an ensemble. See \Cref{thm:dist-implies-q.resources}. That is, these resources become necessary well before the ensemble can be fully pseudorandom. To quantify the `imaginarity' and coherence of a unitary $U$, we use the definitions of Ref.\ \cite[Section 4]{haug_pseudorandom_2024}, which quantify the imaginarity and coherence respectively of the corresponding Choi state of $U$ as 
$I_p(U) := 1-\frac{1}{N^2}\bigl|\tr[U^\dagger U^*]\bigr|^2$ and $C_p(U) := -\frac{1}{N}\sum_{x,y=0}^{N-1} |U_{x, y}|^2 \ln |U_{x ,y}|^2.
$

\begin{theorem}[Distinctness implies quantum resources {(\Cref{thm:imag-necc-for-dist} and \Cref{thm:coh-necc-for-dist})}]\label{thm:dist-implies-q.resources}
 Let a unitary ensemble $\cal E$ be $\delta$-distinct. 
Then,
\begin{equation}
\textup{Imaginarity:}\hspace{2mm}\E_{U\leftarrow\calE}[I_p(U)] \geq 1-\delta
\qquad\text{and}\qquad
\textup{Coherence:}\hspace{2mm}\E_{U\leftarrow\calE}[C_p(U)] \geq \ln\left(1/\delta\right).
\end{equation}
\end{theorem}
Noting that PRUs must be $\negl(n)$-distinct (\Cref{thm:PRU-must-be-negl-dist}), we recover resource Theorems 3 and 5 for PRUs in Ref.\ \cite{haug_pseudorandom_2024} already at the distinct subspace level. Speaking of resources, it is natural to ask how does entanglement of the unitary ensemble relate to distinctness. We show that a single layer of  single-qubit random Cliffords is $\negl(n)$-distinct (\Cref{prop:single-layer-cliffords-eac}). And this suffices to replace the $\log(n)$ depth random Clifford (or $n$-qubit unitary $2$-design) $C$, in the $PFC$ ensemble with a single layer of single qubit 2-designs, e.g., single qubit random Cliffords.

\paragraph{Is distinctness just a state 2-design in disguise?} Having established that distinctness is indeed necessary for PRUs, we next ask whether it is genuinely weaker than the $2$-design property used to obtain distinctness in all existing constructions \cite{metger_simple_2024,schuster_strong_2025,ma_how_2024, lu_parallel_2025}. In other words, 
are there distinct ensembles that fail to be both unitary and state designs? We answer this question in the affirmative by constructing an explicit ensemble that $\mathcal{O}(t^2/2^n)$-distinct, yet it fails to even be a {\em state} $1$-design! We stress that not only is $\mathcal{O}(t^2/2^n)$-distinctness much stronger than $\negl(n)$-distinctness, it is exactly the same distinctness that the existing constructions \cite{metger_simple_2024,schuster_strong_2025,ma_how_2024, lu_parallel_2025} achieve using a state or a unitary 2-design, whereas our ensemble is not even a state 1-design.   Our ensemble is a {\em complex} random phase unitary, $F_{\C}$\footnote{Notation: Unless clear from context, we use the $X_{\C}$ or $X_{\R}$ to denote whether the operator $X$ is complex or real respectively.}, followed by $n$-qubit Hadamard, $H$. Henceforth, the $HF_{\C}$ ensemble.

\begin{theorem}[Distinctness of $HF$ ensemble without being a state design]\label{thm:HF-dist-but-not-state-2-design}
    Let $H$ be the $n$-qubit Hadamard transform and $F \coloneqq \sum_{x \in [N]} \omega^{f(x)}\ketbra{x}{x}$, where $\omega = e^{2\pi i/3}$, and $f: \zo^n \rightarrow \{0, 1, 2\}$ is a random 4-wise independent ternary function. Then, $HF$ is $\mathcal{O}(t^2/N)$-distinct, yet it fails to be a state 1-design.
\end{theorem}
\Cref{lemma:HF-distinct} shows that when $f$ is uniformly random ternary function, which interestingly also appears in strong PRU constructions \cite{ma_how_2024,schuster_strong_2025}, $HF$ ensemble is $(2/N)$-EAC, meaning that $HF$ is $\mathcal{O}(t^2/N)$-distinct by \Cref{lemma:EAC-implies-distinctness}. Recall that since $N=2^n$, $\mathcal{O}(t^2/N)$-distinctness is much stronger than $\negl(n)$-distinctness, which is what we need for PRUs. Since we only ever use a 2-copy $F$-twirl in our proof of \Cref{lemma:HF-distinct}, we instead sample the random ternary $f$ from a $4$-wise independent ternary function family. This is because the results of Ref.\ \cite[Theorem 3.1]{zhandry2015secure} guarantee that any quantum algorithm making at most $t$ queries to a uniformly random function acts identically if we replace the uniformly random function by a $2t$-wise independent function. 
\begin{remark}[$HF$ is not a state $1$-design]\label{remark:HF-not-1-design}
 Despite being $(2/N)$-EAC, $HF$ ensemble is not even a state $1$-design. This follows by a simple observation that $HF \ket{0} = \ket{+}$, which is $\Omega(1)$ far (in trace distance) from the single copy Haar random state, i.e., the maximally mixed state.    
\end{remark}

The $HF$ ensemble might also be independently interesting because standard anticoncentration has long been conjectured to be equivalent to state 2-designs. Authors of Ref.\  \cite{heinrich_anti-concentration_2025} show that this is indeed true on locally invariant architectures e.g., a brickwork random quantum circuit. Anticoncentration, yet failure to be a state $2$-design, of the $HF$ ensemble shows that this equivalence must be architecture dependent. 

\subsubsection{Applications of distinctness}
\label{subsubsec:apps-of-distinctness}
After having isolated distinctness as the relevant threshold for PRUs, we use it in two complementary directions. Constructively, it provides
the minimal statistical input required by 
the $PFC$-type constructions (\Cref{thm:negl-distinct-to-PRU}), allowing
design-level statistical randomization to be replaced by substantially simpler ensembles. Obstructively, every failure of `sufficient' distinctness yields a collision-based distinguisher 
and hence a no-go test for pseudorandomness. The same principle underlies the quantum resource lower bounds for PRUs as we saw in \Cref{thm:dist-implies-q.resources}. We finally show that these resource obstructions become input dependent when pseudorandomness is required only on restricted classes (allowed to the adversary/distinguisher to query the unknown unitary on) of input states. We characterize these states algebraically and also propose concrete (and physically relevant) ensembles of such states.

\paragraph{New PRU ensembles from new distinct ensembles.}

We start by observing that the $PF$ ensemble (with a pseudorandom function and permutation) yields a non-adaptive PRU on the distinct subspace \cite[Theorem 5.2]{metger_simple_2024}. Meaning that if an ensemble of unitaries $D$ is distinct, then $PFD$ is a non-adaptive PRU. In this section, we will focus on new distinct ensembles $D$ that together with $PF$, yield $PFD$ as non-adaptive PRU. Recall that $D$ is generally taken be a unitary 2-design that ensures $\mathcal{O}(t^2/N)$-distinctness \cite[Lemma 3.2]{metger_simple_2024}.

\begin{theorem}[{Implicit in \cite[Section~3]{metger_simple_2024}}]\label{thm:negl-distinct-to-PRU}
If an $n$-qubit unitary ensemble $D$ is
$\delta$-distinct, then
\begin{equation}
    \left\|\calM_{PFD}^{(t)}-\calM_{U_{\Haar}}^{(t)}\right\|_\diamond
    \leq O\left(\sqrt{\delta}+\frac{t^2}{N}\right)
\end{equation}
for $t\ll N\coloneqq 2^n$.
\end{theorem} 
One can then replace the random function $F$ and the random permutation $P$ by their suitable pseudorandom analogues in order to get the desired version of pseudorandomness: statistical or computational \cite{metger_simple_2024}. Replacing the random permutation $P$ and the random function $F$ by their quantum computationally secure pseudorandom counterparts and setting $\delta = \negl(n)$ in \Cref{thm:negl-distinct-to-PRU} implies a neglibile diamond norm, yielding a non-adaptive PRU. Observing that $\negl(n)$-distinct ensemble $D$ already suffices to make $PFD$ a PRU (\Cref{thm:negl-distinct-to-PRU}), we construct an explicit $\negl(n)$-distinct ensemble in unit depth. Concretely, we show that a single layer of random single qubit Cliffords is $(2/3)^n$-EAC (\Cref{prop:single-layer-cliffords-eac}), which for $t=\poly(n)$ implies $\negl(n)$-distinctness for the ensemble.

Existing constructions use a unitary $2$-design to get $\mathcal{O}(t^2/N)$-distinctness. That route incurs an unavoidable depth of at least $\log(n)$. \Cref{prop:single-layer-cliffords-eac} brings this depth down to just $1$, which together with \Cref{thm:negl-distinct-to-PRU} gives a $\log(n)$ depth saving to instantiate the distinctness generating unitary in the $PFC$ ensemble against computationally bounded adversaries. 
\begin{corollary}\label{corr:PFC-wth-1-qubit-Cliffords}
    Let $\bigotimes_{i=1}^n C_i$ be a layer of single-qubit Cliffords, where each $C_i$ is drawn uniformly randomly from $1$-qubit Clifford group. Then, we have that $PF\bigotimes_{i=1}^n C_i $
 is a PRU.
\end{corollary}

Nevertheless, if we desire higher order $t$-designs for $t \ll N$, we can use the $\mathcal{O}(t^2/N)$-distinctness of the $HF_{\C}$ ensemble (\Cref{thm:HF-dist-but-not-state-2-design}) to stitch together a new $t$-design/PRU construction, following the design principle of \Cref{thm:negl-distinct-to-PRU}.

\begin{theorem}\label{thm:PFHF-PRU}
    Let $H$ be the $n$-qubit Hadamard transform and $F_{\C} \coloneqq \sum_{x \in [N]} \omega^{f_1(x)}\ketbra{x}{x}$, where $\omega = e^{2\pi i/3}$, and $f_1: \zo^n \rightarrow \{0, 1, 2\}$ is a random 4-wise independent ternary function. Similarly, define $F_{\R}\coloneqq \sum_{x \in [N]} (-1)^{f_2(x)}\ketbra{x}{x}$, with $f_2$ being a pseudorandom Boolean function. Then, 
        $PF_{\R}HF_{\C}$
    is a PRU.
\end{theorem}

Interestingly, as we discuss later, if the right-most complex-valued $F_{\C}$ is replaced by a real-valued binary phase operator $F_{\R}$, then  $PF_{\R}HF_{\R}$ forms a real-valued PRU on restricted input states. See \Cref{thm:real-PRUs} and \Cref{corr:PFHF-real-PRU}. 

\paragraph{Distinctness as a no-go test for PRUs.}

Pseudorandomness of the alternating phase and Hadamard operators in $PF_{\R}HF_{\C}$ ensemble (\Cref{thm:PFHF-PRU}) reminds us of a structurally similar conjectured PRU construction by Ji, Liu, and Song in the paper \cite{ji_pseudorandom_2018} where they introduced the notion of pseudorandom states and unitaries.
\begin{conjecture}[JLS conjecture {\cite[Section 6.2]{ji_pseudorandom_2018}}]\label{conj:JLS-alt-phase-Hadamard}
    Let $n\in\N$, $l=\calO(1)$, $N=2^n$, $\omega_N=e^{2\pi i/N}$, and, for each $j\in[l]$, let $F_j=\sum_{x\in\{0,1\}^n}
        \omega_N^{f_j(x)}\ket{x}\bra{x}$,
    where each $f_j\colon[N]\to[N]$ is an independently keyed pseudorandom
    functions. Then the ensemble of
    unitaries
    \begin{equation}
        U=F_lH\cdots F_1H
    \end{equation}
    is a PRU ensemble.
\end{conjecture}

We thus ask: can we replace the trailing permutation $P$ in $PF_{\R}HF_{\C}$ ensemble by alternating $\poly(n)$-many i.i.d. phase and Hadamard operators, while maintaining a similar security guarantee? If so, this would resolve the JLS conjecture \Cref{conj:JLS-alt-phase-Hadamard}. We show that this is not possible even if we use polynomially many i.i.d. functions $f:[N] \to [K]$, where $K \leq N/n^{\omega(1)}$, the codomain size is superpolynomially smaller than the domain. Indeed, we get a lower bound on the size of the co-domain of $f$ for \Cref{conj:JLS-alt-phase-Hadamard} to be true. 
\begin{theorem}
    Let $n, k\in\N$, $K\coloneqq 2^k \leq N\coloneqq 2^n$.
    \Cref{conj:JLS-alt-phase-Hadamard} can only hold for functions 
    \begin{equation}f_j\colon[N]\to[K]\end{equation} if  $\log K > n - \omega(\log n)$.
\end{theorem}
See \Cref{thm:two-copy-suffices} for a formal statement. Our proof gives a finer analysis of a simple distinguisher that applies the unknown unitary to both halves of a Bell state and then projects the output back onto that Bell state. This tests for failure of distinctness because the Bell-state projector is supported on the complement of the two-copy distinct subspace.

\paragraph{Input-dependent distinctness and real-valued PRUs.}Recall that for any $\delta$-distinct ensemble we quantified its imaginarity to be at least $1-\delta$ (\Cref{thm:imag-necc-for-dist}). This enforces any distinct ensemble to necessarily be complex-valued. It turns out that we can avoid the need for complex numbers altogether by asking for distinctness on specific input states. To understand the structure of these special input states, we revisit the lower bound in \Cref{thm:imag-necc-for-dist} on the complex resources (imaginarity) of an ensemble to be distinct and observe that this lower bound is witnessed by the maximally entangled input state. This lower bound, as it turns out, is inherently input state dependent. Indeed, we show that the complex resource lower bound for distinctness vanishes if we consider input states that have sufficiently small overlap with the Bell state. We formalize this condition on the input states using a natural metric, which we call the Bell overlap (\Cref{def:bell-overlap}). For a state ensemble $\calS$ and any $t$-copy input state, $\rho_{\X_1, \ldots, \X_t} \in \calS$, Bell overlap measures the maximum trace overlap over all two-copy reduced marginals of $\rho_{\X_1, \ldots, \X_t}$ with the unnormalized Bell state, $\ketbra{\Omega}{\Omega}$\footnote{Note that Bell overlap might appear to be a metric that is low for less entangled states. This is not the case. It only measures trace overlap with the Bell state, not any other (maximally) entangled state. In fact there exists maximally entangled states for which Bell overlap is $0$. See \Cref{remark:compare-grevink}.}. We show that a fully real ensemble (in this case, the {\em real} Clifford ensemble \cite{hashagen_real_2018}, or generally any orthogonal $2$-design) which is $\calO(t^2/N)$-distinct on all input states with $\calO(1)$ Bell overlap. We can go up to Bell overlap at most $N/n^{\omega(1)}$ if we only desire $\negl(n)$-distinctness, which in turn enlarges the state class which real PRUs are secure on.

\begin{theorem}\label{thm:real-PRUs}
    Let $C_{\R}$ be an $n$-qubit real random Clifford and $F_{\R} \coloneqq \sum_{x\in[N]}(-1)^{f(x)}\ketbra{x}{x}$,
for uniformly random Boolean function $f:\zo^n\to\zo$. Further, let $\calS$ be an ensemble of states with Bell overlap at most $N/n^{\omega(1)}$. Then, the ensembles $PF_{\R}C_{\R}$ and $PF_{\R}HF_{\R}$ both form PRUs on all input states in $\calS$. 
\end{theorem}

We use the $PFC_{\R}$ ($PF_{\R}HF_{\R}$ can equivalently be used) ensemble to show that in general, orthogonal and unitary designs are equivalent on input states with constant Bell overlap with error $\calO(t/\sqrt{N})$. This improves the best known to authors' knowledge bound known from Ref.\ \cite{grevink_will_2025}. Many states of interest are already captured by {\em constant} Bell overlap. This includes all PPT (\emph{positive partial transpose}) states, and a large class of NPT (\emph{negative partial transpose}) states, e.g., Choi states of traceless unitaries. See \Cref{remark:compare-grevink} for details. This goes beyond PPT barrier on the input states identified as one of the future directions left open in Ref.\ \cite{grevink_will_2025}. NPT states are necessarily entangled since negativity is a measure of entanglement. Thus NPT states with small Bell overlap informs us that the Bell overlap criterion is not a proxy of entanglement in general. It merely measures closeness with {\em one} maximally entangled (Bell) state. 

We can capture another large class of states (on which real PRUs are secure) using the Bell overlap condition. That is, pure states whose Schmidt rank is at most $N/n^{\omega(1)}$ and the same bound for mixed state but on Schmidt number (extension of Schmidt rank to mixed states defined by Terhal and Horodecki \cite[Definition 1]{terhal_schmidt_2000}). Any state with Schmidt number $k$ has Bell overlap at most $k$ \Cref{thm:SN-imply-BO}. For PRUs, we can tolerate input states with Bell overlap at most $N/n^{\omega(1)}$. This includes any $t$-partite state whose two-copy reduced marginal has Schmidt number at most $N/n^{\omega(1)}$, of which pure bipartite states with Schmidt rank at most $N/n^{\omega(1)}$ are a special case.
In fact, the set of states with Schmidt number $1$ already covers all product and separable states, which resolves the open question Brakerski and Magrafta \cite{brakerski_real-valued_2024} who have asked for real PRUs on product input states.

\subsection{Related work}

\paragraph{Distinctness in PRU constructions.}
Distinct subspaces occur as intermediate objects in the analyses of $PFC$, its (adaptive) variants, and parallel Kac's walk construction for PRUs \cite{metger_simple_2024,ma_how_2024,lu_parallel_2025,foxman_quantum_2026, schuster_strong_2025, cui_unitary_2025}. These works establish distinctness using a stronger randomization property, unitary or state 2-designs, and then exploit it inside the security proof. We instead isolate distinctness as a property of an ensemble, prove that it is necessary for every PRU, and show that the 2-design properties previously used to obtain it are unnecessary. The recent $PC$ construction \cite{foxman_quantum_2026} makes a complementary simplification to the $PFC$ ensemble: it removes the phase layer $F$  while retaining a full random Clifford, whereas we retain $PF$ but replace the Clifford by ensembles that are not even 1-design. We propose different distinct ensembles $D$, resulting in different $PFD$ ensembles with different security guarantees (computational, statistical, or input state specific) depending on the distinct ensemble $D$ and the input states allowed.

\paragraph{Anticoncentration and designs.}
Standard anticoncentration controls collisions for a fixed product input \cite{hangleiter_anticoncentration_2018,dalzell_random_2022}. Heinrich, Haferkamp, Roth, and Helsen have shown that, under local-unitary invariance, such anticoncentration is equivalent to an appropriate relative-error state $2$-design condition \cite{heinrich_anti-concentration_2025}. Our $HF_{\C}$ ensemble lies outside this invariant setting and separates the notions more strongly: it matches the unitary $2$-design bound of $\mathcal{O}(2^{-n})$-EAC on arbitrary bipartite inputs while failing to even be a state $1$-design. Operationally, however, $\mathcal{O}(2^{-n})$-EAC guarantee gives $\mathcal{O}(t^2/N)$-distinctness. This is exploited in the collision based distinguisher of Brakerski and Yuen \cite{brakerski_scalable_2026} for the $PFC$ ensemble when $t=\Theta({\sqrt{N}})$. This is due to the fact that the random $\log(n)$ depth Clifford is $\mathcal{O}(1/N)$-EAC and consequently $\mathcal{O}(t^2/N)$ distinct. This coupled with a square root loss due to the gentle measurement lemma \cite[Lemma 9]{winter_coding_1999} in the $PFC$ analysis \cite{metger_simple_2024} gives security in the regime  $t \ll \sqrt{N}$. While our distinct ensembles also do not address the case for when $\sqrt{N} \leq t \leq N$, we switch perspectives towards computational security and give depth 1 distinct ensemble, as opposed to $\log(n)$-depth Clifford, that is secure so long as $t=\poly(n)$. That is, whenever the distinguisher is computationally bounded. 

\paragraph{The alternating phase-Hadamard route to PRUs.}
Ji, Liu, and Song have proposed that a constant number of independently 
keyed phase--Hadamard layers
\begin{equation}
F_\ell H\cdots F_1H,
\qquad
F_j=\sum_{x\in[N]}\omega_N^{f_j(x)}\ketbra{x}{x},
\qquad
f_j\colon[N]\to[N],
\end{equation}
with $\omega_N = e^{2\pi i/N}$
should form a PRU \cite[Section~6.2]{ji_pseudorandom_2018}. Their proposal has been motivated in part by the design-theoretic result of Nakata et al.\  \cite{nakata_unitary_2017, nakata_efficient_2017}, who have shown that alternating random unitaries diagonal in complementary bases approach a unitary $t$-design. At the 
same time, for \emph{pseudorandom quantum states} (PRSs),
the authors of Ref.\ \cite{ji_pseudorandom_2018} have shown that $F_1H$ is already statistically indistinguishable from Haar random states and conjectured the security for the real {\em binary} phase states, where $f_j: \zo^{n}\to \zo$ is a Boolean function and $\omega_N=-1$. This has 
actually been resolved in the affirmative by Brakerski and Shmueli \cite{brakerski_pseudo_2019}. Our lower bound on the co-domain size of $f_j$ exposes a sharp gap between PRSs and PRUs. Even polynomially many independent phase-Hadamard layers cannot form a PRU when the phases are generated from functions $f_j\colon[N]\to[K]$ satisfying
\begin{equation}
\log K\leq n-\omega(\log n), \text{ equivalently }
K\leq {N}/{n^{\omega(1)}}.
\end{equation}
The binary case is already obstructed by the impossibility of real PRUs \cite{haug_pseudorandom_2024}. Our result also rules out ternary functions and, more generally, genuinely complex phase operators, with functions having co-domains size superpolynomially smaller than $2^n$. Thus, the permutation $P$ in our $PF_{\R}HF_{\C}$ construction cannot be replaced by polynomially many phase--Hadamard layers, so long as $K\leq {N}/{n^{\omega(1)}}$. This does not refute the original JLS conjecture, which uses the full alphabet $K=N$. Bostanci, Haferkamp, Hangleiter, and Poremba have more recently proposed Hamiltonian phase states based instantiation of the similar (phase-Hadamard type) PRU \cite[Section 6.6]{bostanci_efficient_2024}. Its 
security, however, remains 
conjectural (and some of it will be discussed in upcoming sRef.\ \cite{HPS_implies_OWF}).

\paragraph{Quantum resources and real PRUs.}
Haug, Bharti, and Koh have  derived imaginarity and relative entropy of coherence requirements directly from PRU security \cite{haug_pseudorandom_2024}. We have shown that both requirements already follow from distinctness. On the other hand, Brakerski 
and Yuen have proven a different entropy obstruction
\cite{brakerski_scalable_2026}. They count how many unitaries an {\em ensemble} must have in order to be a $t$-design. These authors show that diagonal unitaries, although 
they form an uncountable
set, can be replaced up to approximation by classical functions. 
Our relative entropy of coherence lower bound is robust in the sense that it also applies to distinctness (and hence pseudorandomness) of unitary ensembles beyond purely diagonal matrices, i.e., with small off-diagonal entries. On the other hand, Brakerski and Magrafta obtain real pseudorandomness on polynomially many mutually orthogonal inputs \cite{brakerski_real-valued_2024}, whereas our Bell-overlap condition far subsumes this class of states by ensuring security on states, e.g., arbitrary product states and even separable states, that are not necessarily orthogonal. All such states were assumed to be a proxy for positivity of the partial transpose of the input state. We show input states with negative partial transpose on which real PRUs (like $PFC_{\R}$) are still secure. We go even further and use our real $PFC_{\R}$ ensemble to improve the general equivalence of unitary and orthogonal twirl  devised in Ref.\ \cite{grevink_will_2025} devised only on PPT input states to constant Bell overlap states, which not only includes all PPT states, but also many NPT states.

\subsection{Discussion and open questions}

 In this work, we have identified distinctness as a  property that replaces the need for statistical pseudorandomness needed to obtain computational pseudorandomness in the case of non-adaptive forward query PRUs.  We used this `lens of distinctness' to understand various properties of PRUs: new constructions and their limitations, quantum resource constraints (e.g., on imaginarity), and when do such resource constraints vanish when we restrict the input states that can be used by the distinguisher. There are several ways to further develop this program.

\paragraph{Distinctness in stronger query models and beyond.}
A natural next step is to extend the theory of distinctness to stronger query models, including inverse, transpose, conjugate, controlled, and adaptive access. Is there a hierarchy of distinctness conditions corresponding to these increasingly powerful notions of security? Can such conditions again be separated from unitary-design properties and used both to construct and to rule out PRUs and $t$-designs secure against the corresponding query models? A satisfactory theory should distinguish the statistical properties genuinely required by each form of oracle access from those that arise only as artifacts of current proof techniques.
Beyond PRUs, we expect that in applications where collision suppression is the principal requirement, distinct ensembles can be substantially cheaper to realize than unitary $2$-designs. Recently, Ref.~\cite{foxman_quantum_2026} introduced a ``non-plussed'' version of the distinct subspace, which additionally has no support on $\ketbra{+}{+}^{\ot t}$. In 
the language of the present work, a unitary $2$-design remains $\mathcal{O}(t^2/N)$-distinct with respect to this strengthened subspace. This observation is used in Ref.~\cite{foxman_quantum_2026} to show that $PC$ is a PRU, thereby removing the random-function layer $F$ from the $PFC$ construction.
In the light of this, it
would be valuable to determine whether the distinct ensembles constructed here are also non-plussed distinct, and whether the corresponding PRU constructions remain secure after removing the binary phase operator $F$. 
In particular, is the depth-one ensemble of independent single-qubit Clifford gates non-plussed distinct, and would this imply that
$P\bigotimes_{i=1}^{n} C_i$
is a PRU? An affirmative answer would arguably yield one of the simplest PRU constructions known.

\paragraph{Real-valued pseudorandomness and applications.}
Another similar question is whether the real- and complex-valued PRUs and $t$-designs proposed here remain secure against stronger forms of query accesses. For real-valued pseudorandom constructions, it would also be useful to identify interesting and, ideally, physically relevant ensembles of input states with bounded Bell overlap. This would imply security of real unitary designs and PRUs, such as $PFC_{\R}$, on substantially larger classes of input states. Our connection with Schmidt numbers provides a natural starting point for such a classification.
Our results show, in particular, that the obstruction to real-valued pseudorandomness is not entanglement by itself, but sufficiently large overlap with the canonical maximally entangled Bell state. 

Moreover, the security of real-valued PRUs on states of small Bell overlap implies that a distinguisher between real- and complex-valued PRUs can serve as a witness of non-trivial Bell overlap in the distinguisher's query state. It would be interesting to develop this observation into a certification or resource-detection protocol beyond the immediate setting of pseudorandomness.

\paragraph{Distinctness beyond PRUs.}
Our results suggest a broader programme of replacing design conditions by weaker operational properties tailored to the task at hand. 
This perspective may also be useful beyond pseudorandom unitaries. Distinct ensembles may provide cheaper replacements for unitary designs in randomized measurements, benchmarking, sampling, and cryptographic protocols in which repeated output labels constitute the relevant failure mode. Understanding which of these applications genuinely require design-level randomness, and which require only collision suppression or related operational properties, is an interesting direction for future work.

\paragraph{Outlook.}
The theory of quantum pseudorandomness has largely developed through two complementary languages. Statistical notions such as unitary designs, quantify how closely an ensemble resembles Haar randomness, whereas computational pseudorandomness characterize what efficient observers can distinguish. Distinctness brings these perspectives closer together. 
It is our hope that this specific viewpoint will help replace unnecessarily strong conditions by sharper, and ideally, operational criteria.
More broadly, identifying such intermediate notions may lead to a `modular' theory of quantum pseudorandomness.
Such a theory would 
not only simplify constructions, but could also help identify which experimentally accessible families of quantum states and unitaries possess precisely the randomness required for cryptographic applications.

\subsection{Acknowledgements} 
We thank Lennart Bittel, John Bostanci, Lorenzo Grevink, Jonas Haferkamp, Tobias Haug, and Jonas Helsen for helpful discussions and useful feedback on the manuscript. B.F.\ acknowledges support from AFOSR (FA9550-21-1-0008 and
FA9550-26-1-B214). This
material is based upon work partially supported by the National
Science Foundation under
Grant CCF-2044923 (CAREER). The Berlin team 
acknowledges funding by the BMFTR (Hybrid++,
MuniQC-Atoms), the Munich Quantum Valley, Berlin Quantum, the 
Quantum Flagship (Millenion, PasQuans2), the European Research Council (DebuQC),
the Clusters of Excellence (MATH+, ML4Q), and the 
DFG (CRC 183, SPP 2514, and BoLaCo).

\subsection*{AI statement} 
The authors started working on this project in early 2025 and obtained preliminary versions of the main results, namely constructions of distinct ensembles, their properties, and their applications to pseudorandomness, with the notable exception of \Cref{prop:single-layer-cliffords-eac}, which was proposed by ChatGPT 5.6 Sol as a counter-example to authors' conjecture that a $\negl(n)$-distinct ensemble must be entangling. In addition to literature search and assistance with exposition and verifying technical correctness, ChatGPT 5.5 and 5.6 Pro were used to devise proof strategies for all the main results. Authors independently verified all the proofs and take full responsibility for the content. 

\section{Preliminaries}
\paragraph{Notation.}
We denote the number of qubits by $n \in \mathbb{N}$ and set
$N := 2^n$, so that each quantum register is identified with the  complex Hilbert space $\mathbb{C}^N$ of dimension $N$.  The computational basis $\zo^n$ are indexed by $[N] := \{1,\ldots,N\}$. We write $\calL(\C^N)$ for the space of linear operators on $\C^N$. For every linear operator $X \in \calL(\C^N)$, we define the Schatten $p$-norm of $X$ by  $\|X\|_p := \left( \operatorname{Tr}\left[ |X|^p \right] \right)^{1/p}$, 
for all $p \in [1,\infty]$, where $|X| := \sqrt{X^\dagger X}$. We represent $n$-qubit Hadamard transform $H^{\ot n}$ by $H$ and transpose as $\mathsf{T}$.
The parameter $t=t(n)$ denotes the number of oracle queries. For the purposes of this paper, we are primarily interested in the regime when $t$ is polynomially bounded in $n$. Throughout, sans-serif capital letters denote quantum registers. We write $\mathsf{X}_1,\ldots,\mathsf{X}_t$ for $t$ registers, each associated with a Hilbert space isomorphic to $\mathbb{C}^N$, and $\mathsf{R}$ for an arbitrarily large auxiliary register. Whenever obvious from context, we will not explicitly write registers on which an operator is supported in order to simplify presentation. We further define
\begin{equation}
[N]^t_{\dist}
:=
\left\{
(x_1,\ldots,x_t)\in [N]^t : x_i \neq x_j \text{ for all } i\neq j
\right\},
\end{equation}
to be the set of all pairwise distinct $t$-tuples. The orthogonal projector onto the corresponding distinct subspace is
\begin{equation}
\Pi^{\dist}
:=
\sum_{(x_1,\ldots,x_t)\in [N]^t_{\dist}}
|x_1,\ldots,x_t\rangle\langle x_1,\ldots,x_t|,
\end{equation}
and we denote the complementary projector by
\begin{equation}
\overline{\Pi}^{\dist}
:=
\Id_{\mathsf{X}_1,\ldots, \mathsf{X}_t}
-
\Pi^{\dist},
\end{equation}
which can be thought of `detecting a collision' on the full $t$-copy Hilbert space. An `equality projector' detects a collision on a given bipartition.
\begin{equation}
{\Pi}^{\eq}_{\X_i, \X_j} = \sum_{x\in [N]} \ketbra{x}{x}_{\X_i} \otimes \ketbra{x}{x}_{\X_j},
\end{equation}
and acts as identity on the remaining registers. $\overline{\Pi}^{\dist}$ and ${\Pi}^{\eq}_{\X_i, \X_j}$ are related by an operator inequality (\Cref{fact:dist-eq-op-ineq})

We write $\mu_{\Haar}$ for the normalized Haar measure on the relevant compact group. Depending on context, this denotes the normalized Haar measure on $\mathrm{U}(N)$ or on $\mathrm{O}(N)$.
We denote by $\Cl_{\C}(N)$  (resp. $\Cl_{\R}(N)$) the complex (resp. real) Clifford group of dimension $N$. 
Given a finite set $\calE$, we write $x \leftarrow \calE$ to mean that $x$ is sampled uniformly from $\calE$. More generally, if $\mu$ is a probability measure on a measurable space $X$, then $x \leftarrow \mu$ means that $x$ is sampled according to $\mu$.
The following two operators 
\begin{equation}
|\Omega\rangle := \sum_{x\in [N]} |x,x\rangle \qquad \SWAP
=
\sum_{x,y\in[N]}
\ket{x,y}\bra{y,x}
\end{equation}
on $\calH^{\ot 2}$ shall appear multiple 
times in this work.
The unnormalized maximally entangled state 
on two $N$-dimensional registers is denoted as $\ket{\Omega}$ and normalized maximally entangled state vector  $\ket{\Omega}/\sqrt{N}$ is denoted by $\ket{\Phi}$.

Let $\Sym_N$ be the symmetric group on $N$ elements. Then for $\pi \in \Sym_N$,
\begin{equation}
P_{\pi} := \sum_{x \in [N]} \ketbra{\pi(x)}{x}.
\end{equation}
We will omit the dependence on $\pi$ and write $P$ to simplify notation. Unless otherwise specified, $P$ corresponds to $P_{\pi}$, where $\pi \leftarrow \Sym_N$. For a diagonal phase operator in the computational basis, we write
\begin{equation}
F_j=\sum_{x\in[N]}\omega_q^{\,f_j(x)}\ketbra{x}{x},\qquad f_j:[N]\to \{0,1,\ldots, K-1\},
\end{equation}
where $\omega_q\coloneqq e^{2\pi i/q}$ is the $q$-th root of unity for some $q\geq2$ and $f_j$ is sampled uniformly randomly from the set of all functions $[N] \to \{0,1,\ldots, K-1\}$, where $K \leq N$. To avoid clutter, we will simply write $F$, which is to be understood as a diagonal phase operator with complex phases. Whenever needed, we will make the parameters $q$ and $K$ explicit.

We write $\negl(n)$ for negligible functions. That is any function that is $o(1/n^c)$ for all $c>0$.
Unless stated otherwise, we only consider PRUs against adversaries that have \emph{forward-only} oracle access and
can query the oracle \emph{only in parallel} (non-adaptively). Our goal is to understand the necessary resources
for constructing PRUs, so we work with this weakest form of security. Any lower bound or impossibility result
proved in this regime applies \emph{a fortiori} to stronger notions of PRU security (for example, allowing
adaptive queries or inverse-oracle access).
\begin{definition}[Parallel (forward) pseudorandom unitaries]\label{def:parallel-pru}
We say $\{\mathcal{U}_n\}_{n\in\mathbb{N}}$ is a secure \emph{parallel forward PRU} if, for all $n\in\mathbb{N}$,
\begin{equation}
\mathcal{U}_n=\{U_k\}_{k\in\mathcal{K}_n}
\end{equation}
is a set of $n$-qubit unitaries (where $\mathcal{K}_n$ denotes the keyspace) satisfying the following properties:
\begin{itemize}
  \item \textbf{Efficient computation:} There exists a $\poly(n)$-time quantum algorithm that implements the
  $n$-qubit unitary $U_k$ on input $k\in\mathcal{K}_n$.

  \item \textbf{Parallel forward-query indistinguishability from Haar:}
  For any QPT algorithm $\calA$, whose oracle access is \emph{forward-only} and \emph{non-adaptive}, and that measures a two-outcome observable $D_{\X\mathsf{R}}$
  with eigenvalues $\{0,1\}$ after the queries, we have
  \begin{equation}
  \left|
  \mathbb{E}_{\mathcal{O}\leftarrow \mathcal{U}_n}\,
  \tr\left(D_{\X,\mathsf{R}} \cdot \left|\calA^{\mathcal{O}}\right\rangle\left\langle \calA^{\mathcal{O}}\right|_{\X\mathsf{R}}\right)
  -
  \mathbb{E}_{\mathcal{O}\leftarrow \mu_{\Haar}}\,
  \tr\left(D_{\X,\mathsf{R}} \cdot \left|\calA^{\mathcal{O}}\right\rangle\left\langle \calA^{\mathcal{O}}\right|_{\X,\mathsf{R}}\right)
  \right|
  \le \negl(n),
  \end{equation}
  where
 $\left|\calA^{\mathcal{O}}\right\rangle \coloneqq \calA_{1, \X,\mathsf{R}} \big(\mathcal{O}^{\otimes q(n)}_{\X} \otimes \Id_{\mathsf{R}} \big) \calA_{0, \X,\mathsf{R}}\ket{\psi_0}_{\X,\mathsf{R}}$, for any $q(n)=\poly(n)$, where $\mathsf{X}$ is a $q(n)$ qubit query register, $\mathsf{R}$ is an arbitrarily large reference register, $\calA_{0, \X,\mathsf{R}}\ket{\psi_0}, \calA_{1, \X,\mathsf{R}}$ are arbitrary unitaries on the joint register $\X,\mathsf{R}$, and $\ket{\psi_0}_{\X,\mathsf{R}}$ is some efficiently preparable quantum state vector.
\end{itemize}
\end{definition}
We will also consider ensembles that form PRUs on a set of input states $\cal S$, meaning that the QPT distinguisher is only allowed to query the unknown unitary on states in $\cal S$.
Whenever we say $\cal E$ is a PRU ensemble we mean a PRU in the sense of \Cref{def:parallel-pru}, unless otherwise specified.

\paragraph{Representation theory background.}
Let $\operatorname{GL}(N)$ be the group of invertible complex $N\times N$ matrices, equivalently the group of invertible linear maps on $\calH$. We will use the following notions.
\begin{definition}[Representation]
Let $G$ be a group. A representation of $G$ 
on $\calH$ is a group homomorphism
$R:G\to \operatorname{GL}(N)$. Given such a representation $R$, the associated $t$-fold tensor representation acts on $\calH^{\otimes t}$ by
\begin{equation}
g \mapsto R(g)^{\otimes t},
\qquad g\in G.
\end{equation}
If $R$ is a representation, then $g\mapsto R(g)^{\otimes t}$ is also a representation of $G$.
\end{definition}

\begin{definition}[$t$-th order commutant]\label{def:t-th-order-commutant}
Let $R$ be a representation of a group $G$ on $\calH$. The $t$-th order commutant of $G$ with respect to $R$, denoted $\Comm(G, t)$, is the subspace of $\calL\left((\C^N)^{\otimes t}\right)$ given by
\begin{equation}
\Comm(G, t)
=
\bigl\{
A\in \calL\left((\C^N)^{\otimes t}\right)
\big|
[A,R(g)^{\otimes t}]=0 \text{ for all } g\in G
\bigr\}.
\end{equation}
\end{definition}

We will primarily work with the unitary and the orthogonal groups of dimension $N$, which we denote as $\rmU(N)$ and $\rmO(N)$, respectively. 
We interchangably use the terms $t$-wise twirl and $t$-th moment operator.
\begin{definition}[$t$-wise twirl]
Let $\calE$ be an ensemble of unitary operators. The associated $t$-fold twirling channel, also called the $t$-th moment operator, is the linear map
\begin{equation}
\mathcal{M}_{\calE}^{(t)}(X)
:=
\mathbb{E}_{U \leftarrow \calE}\left[
U^{\otimes t} X U^{\dagger,\otimes t}
\right],
\qquad
X \in \mathcal{L}\bigl((\mathbb{C}^N)^{\otimes t}\bigr).
\end{equation}
If $U$ is distributed according to Haar measure on $\mathrm{U}(d)$, we write
\begin{equation}
\mathcal{M}_{U_{\Haar}}^{(t)}(X)
:=
\mathbb{E}_{U \leftarrow \mu_{\Haar}}\left[
U^{\otimes t} X U^{\dagger,\otimes t}
\right].
\end{equation}
\end{definition}
The twirling channel w.r.t.\ the orthgonal group is similarly defined and denoted as $\mathcal{M}_{O_{\Haar}}^{(t)}(X)$.

\begin{definition}[Haar measure]
Let $\rmU(2^n)$ be the group of $n$-qubit unitaries. The Haar measure on $U(2^n)$ is the unique probability measure $\mu_{\mathrm{Haar}}$ on $U(2^n)$ that is invariant under both left and right multiplication, namely
\begin{equation}
\mu_{\mathrm{Haar}}(VS)=\mu_{\mathrm{Haar}}(S)
\quad\text{and}\quad
\mu_{\mathrm{Haar}}(SV)=\mu_{\mathrm{Haar}}(S)
\end{equation}
for every measurable set $S \subseteq U(2^n)$ and every $V \in U(2^n)$.
\end{definition}

For brevity, we will sometimes use $U_{\Haar}$ and $O_{\Haar}$ to denote a Haar random unitary or orthgonal matrix.

\begin{definition}[Unitary $t$-design]
A distribution $\mathcal{D}$ over $n$-qubit unitaries is called a \emph{unitary $t$-design} if
\begin{equation}
    \E_{U \leftarrow \mathcal{D}}\left[ U^{\otimes t}\otimes U^{\dagger,\otimes t}\right]
    =
    \int_{U(2^n)} U^{\otimes t}\otimes U^{\dagger,\otimes t}\, d\mu(U),
\end{equation}
where $\mu$ denotes the Haar measure on $U(2^n)$.
\end{definition}

Since we mostly address necessary conditions for PRUs, the natural choice is the standard notion of an additive error or diamond norm error $t$-design.

\begin{definition}[Additive error unitary $t$-design]
    For $\varepsilon > 0$, an ensemble $\calE$ is an $\varepsilon$-approximate additive error unitary $t$-design if 
    \begin{equation}
\bigl\|
\mathcal{M}_{\calE}^{(t)}-\mathcal{M}_{U_{\Haar}}^{(t)}
\bigr\|_\diamond
\leq
\varepsilon,   
\end{equation}
where $\bigl\|
\mathcal{M}_{\calE}^{(t)}-\mathcal{M}_{\calE'}^{(t)}
\bigr\|_\diamond \coloneqq \max_{\rho}\bigl\| \mathcal{M}_{\calE}^{(t)}(\rho)-\mathcal{M}_{\calE'}^{(t)}(\rho) \bigr\|_1$. The maximization is over all states on $nt$ system qubits with arbitrarily large auxiliary registers.
\end{definition}

An ensemble is called a {\em state} $t$-design if the input states are restricted to be of the form $\rho = \ketbra{\psi}{\psi}^{\ot t}$ for any (normalized) $\ket{\psi} \in (\C^2)^{\ot n}$ \cite{mele_introduction_2024,haah_efficient_2024}. Additive error orthogonal designs 
are similarly defined but w.r.t.\ the 
orthogonal group $\rmO(N)$.

\paragraph{Useful facts.}
Here we recall some well-known facts that we will refer to later.

\begin{fact}[Dirichlet kernel identity]\label{fact:Dirichlet-kernel-identity}
$\sum_{k=0}^{N-1} e^{ikx}
=
e^{i(N-1)x/2}\,\frac{\sin(Nx/2)}{\sin(x/2)}$.
\end{fact}

\begin{fact}[Collision-projector inequalities]\label{fact:dist-eq-op-ineq}
${\Pi}^{\mathsf{eq}}_{\X_i, \X_j} \preceq \overline{\Pi}^{\dist} \preceq \sum_{1\leq i<j\leq t} {\Pi}^{\mathsf{eq}}_{\X_i, \X_j}$
\end{fact}
\begin{proof}
  Since $\overline{\Pi}^{\dist}$ and $\Pi^{\eq}_{\X_i,\X_j}$ are diagonal in the computational basis $(x_1,\ldots,x_t)\in [N]^t$, it suffices to compare diagonal matrix elements on each basis vector, which reduces to the scalar inequality
\begin{equation}
\mathbf{1}[x_i=x_j] \leq  \mathbf{1}\left[(x_1,\ldots,x_t)\notin [N]^t_{\mathrm{dist}}\right]
\leq
\sum_{1\leq i<j\leq t}\mathbf{1}[x_i=x_j].
\end{equation}
The first inequality holds because if $x_i=x_j$, then the tuple is not pairwise distinct. The second inequality holds because a non-distinct tuple contains at least one colliding pair.
\end{proof}

\begin{fact}[Character orthogonality]\label{fact:character-ortho}
For $x, w\in\mathbb{F}_2^n$, define the character $\chi_w:\mathbb{F}_2^n\to\{\pm1\}$
\begin{equation}
\chi_w(x):=(-1)^{\langle x,w\rangle}, \qquad
\langle x,w\rangle := \sum_{i=1}^n x_i w_i \pmod 2.
\end{equation}
Then we have that
$
\mathbb{E}_{x\in\mathbb{F}_2^n}\big[\chi_w(x)\chi_{w'}(x)\big]=\delta_{w,w'}.
$
Equivalently,
$
\sum_{x\in\mathbb{F}_2^n}\chi_w(x)\chi_{w'}(x)
=
N\delta_{w,w'}.
$
\end{fact}

\section{Necessity and characterization of distinctness} \label{sec:distinctness-and-EAC-characterization}
\medskip
\paragraph{Distinctness, entangled anticoncentration and the regime of their equivalence.}

\begin{lemma}[$\delta$-distinctness implies $\delta$-EAC]\label{lemma:distinctness-implies-EAC}
    An $n$-qubit unitary ensemble $\calE$ is $\delta$-EAC if it is $\delta$-distinct.
\end{lemma}
\begin{proof}
    First, let us assume that $\calE$ is $\delta$-distinct (\Cref{def:delta-distinctness}). 
    Then,
    \begin{equation}
\tr\left[
\overline{\Pi}^{\dist}
\E_{U\leftarrow\calE}\left[
U^{\otimes t}\rho_{\X_1\cdots \X_t}U^{\dagger,\otimes t}
\right]
\right]
\le \delta.
\label{eq:distinctness-hypothesis}
\end{equation}
By \Cref{fact:dist-eq-op-ineq}, we have that for any $i\neq j \in [t]$
\begin{equation}
    \Pi^{\eq}_{\X_i, \X_j} \preceq \overline{\Pi}^{\dist}.
\end{equation}
Due to the moment operator being positive semi-definite, $\E_{U\leftarrow \calE}\left[U^{\ot t} \rho_{\X_1\cdots \X_t}U^{\ot t, \dagger}\right]\succeq 0$, 
\begin{equation}
     \tr\left[\Pi^{\eq}_{\X_i, \X_j} \E_{U\leftarrow \calE}\left[U^{\ot t} \rho_{\X_1\cdots \X_t}U^{\ot t, \dagger}\right]\right] \leq \tr\left[
\overline{\Pi}^{\dist}
\E_{U\leftarrow\calE}\left[
U^{\otimes t}\rho_{\X_1\cdots \X_t}U^{\otimes t, \dagger}
\right]
\right].
\label{eq:tr-eq-leq-tr-distinct-comp}
\end{equation}
Then
\begin{align}
\nonumber
     \tr\left[\Pi^{\eq}_{\X_i, \X_j} \E_{U\leftarrow \calE}\left[U^{\ot t} \rho_{\X_1\cdots \X_t}U^{\ot t, \dagger}\right]\right] &= \E_{U\leftarrow \calE}  \tr\left[\left[U_{\X_i}^{\dagger}\ot U_{\X_j}^{\dagger} \Pi^{\eq}_{\X_i, \X_j} U_{\X_i} \ot U_{\X_j}\right] \rho_{\X_1\cdots \X_t}\right] \\
    &= \tr\left[\Pi^{\eq}_{\X_i, \X_j} \E_{U\leftarrow \calE} \left[U_{\X_i} \ot U_{\X_j}\rho_{\X_i, \X_j}U_{\X_i}^{\dagger}\ot U_{\X_j}^{\dagger}\right] \right],
    \label{eq:tr-eq-2-RDM}
\end{align}
where the first step follows by linearity and cyclicity of trace and noting that $\Pi^{\eq}_{\X_i, \X_j}$ acts non-trivially only on registers $\X_i \textnormal{ and } \X_j$. The second step again uses cyclicity and linearity of trace, in addition to the identity $\tr[(M_{\X} \ot \Id_{\X'})Y_{\X, \X'}] = \tr[M_{\X} \tr_{\X'}[Y_{\X, \X'}]]$ that holds for all linear operators $M_{\X}$ and $Y_{\X, \X'}$. Substituting \Cref{eq:tr-eq-2-RDM} in \Cref{eq:tr-eq-leq-tr-distinct-comp}, we get
\begin{equation}
    \tr\left[\Pi^{\eq}_{\X_i, \X_j} \E_{U\leftarrow \calE} \left[U_{\X_i} \ot U_{\X_j}\rho_{\X_i, \X_j}U_{\X_i}^{\dagger}\ot U_{\X_j}^{\dagger}\right] \right] \leq \tr\left[
\overline{\Pi}^{\dist}
\E_{U\leftarrow\calE}\left[
U^{\otimes t}\rho_{\X_1\cdots \X_t}U^{\otimes t, \dagger}
\right]
\right].
\end{equation}
Since $\rho_{\X_1\cdots \X_t}$ is arbitrary choose $\rho_{\X_1\cdots \X_t} = \omega_{\X_i,\X_j}\otimes \tau_{\X_k: k\in[t]\setminus\{i,j\}}$ for any quantum state $ \tau_{\X_k: k\in[t]\setminus\{i,j\}}$. Then, by assumption (\Cref{eq:distinctness-hypothesis}),
\begin{equation}
    \tr\left[\Pi^{\eq}_{\X_i, \X_j} \E_{U\leftarrow \calE} \left[U_{\X_i} \ot U_{\X_j}\omega_{\X_i,\X_j}U_{\X_i}^{\dagger}\ot U_{\X_j}^{\dagger}\right]  \right] \leq \delta,
\end{equation} 
as desired.
\end{proof}

\begin{lemma}[$\delta$-EAC implies $(\delta \cdot t^2)$-distinctness]\label{lemma:EAC-implies-distinctness}
     An $n$-qubit unitary ensemble $\calE$ is $(\delta \cdot t^2)$-distinct if it is $\delta$-EAC.
\end{lemma}
\begin{proof}
    For the reverse direction, the proof closely mimics \cite[Lemma 3.2]{metger_simple_2024}. Starting from the other operator inequality in \Cref{fact:dist-eq-op-ineq},
\begin{equation}
    \overline{\Pi}^{\dist} \preceq \sum_{1\leq i<j\leq t} {\Pi}^{\mathsf{eq}}_{\X_i, \X_j},
\end{equation}
we find that
\begin{align}
     \tr\left[\overline{\Pi}^{\dist}\E_{U \leftarrow \mathcal{E}}U^{\otimes t}\rho_{\X_1\cdots \X_t} U^{\dagger, \otimes t}\right] &\leq \sum_{1\leq i< j \leq t} \tr\left[\Pi_{\mathsf{X}_i, \mathsf{X}_j}^{\eq}\E_{U \leftarrow \mathcal{E}}U^{\otimes t}\rho_{\X_1\cdots \X_t} U^{\dagger, \otimes t}\right]\\
     \nonumber
     &= \sum_{1\leq i < j \leq t} \tr\left[\Pi_{\mathsf{X}_i, \mathsf{X}_j}^{\eq}\E_{U \leftarrow \mathcal{E}}U_{\mathsf{X}_i} \otimes U_{\mathsf{X}_j}\rho_{\mathsf{X}_i, \mathsf{X}_j} U_{\mathsf{X}_i}^{\dagger} \otimes U_{\mathsf{X}_j}^{\dagger}\right]\\
      \nonumber
     &= \frac{t(t-1)}{2} \cdot \tr\left[\Pi_{\mathsf{X}_i, \mathsf{X}_j}^{\eq}\E_{U \leftarrow \mathcal{E}}U_{\mathsf{X}_i} \otimes U_{\mathsf{X}_j}\omega_{\mathsf{X}_i, \mathsf{X}_j} U_{\mathsf{X}_i}^{\dagger} \otimes U_{\mathsf{X}_j}^{\dagger}\right]\\
      \nonumber
     &< t^2 \cdot \tr\left[\Pi_{\mathsf{X}_i, \mathsf{X}_j}^{\eq}\E_{U \leftarrow \mathcal{E}}U_{\mathsf{X}_i} \otimes U_{\mathsf{X}_j}\omega_{\mathsf{X}_i, \mathsf{X}_j} U_{\mathsf{X}_i}^{\dagger} \otimes U_{\mathsf{X}_j}^{\dagger}\right],
      \nonumber
\end{align}
where $\omega_{\mathsf{X}_i, \mathsf{X}_j}$ is an arbitrary state on registers $\X_i, \X_j$. The claim follows by assumption that the ensemble $\calE$ is $\delta$-EAC.
\end{proof}

In the regime where $t$ is polynomial in $n$ and $\delta$ is negligibly small in $n$, combining \Cref{lemma:distinctness-implies-EAC} and \Cref{lemma:EAC-implies-distinctness} yields the equivalence between distinctness and anticoncentration as stated in \Cref{thm:equiv-of-distinctness-and-EAC}.

\paragraph{Necessity of distinctness for PRUs.}

\begin{theorem}[PRUs imply entangled anticoncentration]\label{thm:PRU-negl-EAC}
Any $n$-qubit pseudorandom unitary ensemble $\calE$ must be $\negl(n)$-EAC on all efficiently preparable input states.
\end{theorem}

\begin{proof}
Consider the following distinguisher with oracle access to an unknown unitary oracle $\calO$.
It prepares $\rho_{\X_i, \X_j}$ efficiently, applies $\calO$ to register $\X_i$ and to register $\X_j$ (two parallel oracle calls),
measures both registers in the computational basis obtaining outcomes $x,y\in\{0,1\}^n$,
and outputs $1$ if and only if $x=y$.
Conditioned on a fixed oracle unitary $U$, its acceptance probability over $U\leftarrow \mathcal{E}$ is
\begin{equation}
\Pr[\text{accept}\mid U\leftarrow \mathcal{E}]
=\operatorname{tr}\left[\Pi^{\eq}\,\E_{U\leftarrow \calE}(U\otimes U)\rho_{\X_i, \X_j}(U^\dagger\otimes U^\dagger)\right].
\end{equation}
On the other hand, when U is a Haar random unitary, 
\begin{equation}
\Pr[\text{accept}\mid U\leftarrow \mu_{\Haar}]
=\operatorname{tr}\left[\Pi^{\eq}\,\E_{U\leftarrow \mu_{\Haar}}(U\otimes U)\rho_{\X_i, \X_j}(U^\dagger\otimes U^\dagger)\right],
\end{equation}
where we know that the $2$-wise unitary twirl is a linear combination of identity and $\SWAP$ \cite{mele_introduction_2024}. 
\begin{equation}
\E_{U\leftarrow \mu_{\Haar}}(U\otimes U)\rho_{\X_i, \X_j}(U^\dagger\otimes U^\dagger) = c_{\mathbbm{1}} \mathbbm{1}\ot\mathbbm{1} + c_{\SWAP} \SWAP.
\end{equation}
Thus,
\begin{align}
\nonumber
    \Pr[\text{accept}\mid U\leftarrow \mu_{\Haar}] &= c_{\mathbbm{1}}\tr[\Pi^{\eq}] +  c_{\SWAP} \tr[\Pi^{\eq}\SWAP] \\
     \nonumber
    &= N(c_{\mathbbm{1}}+c_{\SWAP})\\
    &= \frac{N - \tr[\rho_{\X_i, \X_j}\SWAP] + N\tr[\rho_{\X_i, \X_j}\SWAP] -1}{N^2 -1}\\
     \nonumber
    &= \frac{(N-1)(1+\tr[\rho_{\X_i, \X_j}\SWAP])}{(N-1)(N+1)} \\
     \nonumber
    &\leq \frac{2}{N+1},
     \nonumber
\end{align}
where the second step follows because $\tr[\Pi^{\eq}\SWAP] = \tr[\Pi^{\eq}] = N$. In the third step, we import exact expressions for $c_{\mathbbm{1}}$ and $c_{\SWAP}$ \cite[Corollary 13]{mele_introduction_2024}. Last inequality follows by applying \Holder's inequality and using $\norm{\SWAP}_{\infty} = 1$. If, on the other hand, $\calE$ is a PRU ensemble, then by definition (\Cref{def:parallel-pru}): 
\begin{equation}
\big|\Pr[\text{accept}\mid U \leftarrow \calE] - \Pr[\text{accept}\mid U\leftarrow \mu_{\Haar}]\big| \leq \negl(n).
\end{equation}
Thus,
\begin{equation}\Pr[\text{accept}\mid U \leftarrow \calE] \leq \frac{2}{N+1} + \negl(n) = \negl(n).
\end{equation}
\end{proof}

\begin{theorem}[PRUs imply distinctness]\label{thm:PRU-must-be-negl-dist}
    Any $n$-qubit PRU ensemble $\calE$ must be $\negl(n)$-distinct on all efficiently preparable input states.
\end{theorem}
\begin{proof}
    We know from \Cref{thm:PRU-negl-EAC} that $\calE$ must be $\negl(n)$-EAC. Since any efficient distinguisher makes at most $t=\poly(n)$ queries, in that regime \Cref{thm:equiv-of-distinctness-and-EAC} implies that $\negl(n)$-distinctness is necessary for $\negl(n)$-EAC. By transitivity, we get that $\negl(n)$-distinctness is necessary for any PRU ensemble $\calE$. Finally, we only quantify over efficiently preparable input states since that is by definition what a QPT algorithm can query the unknown unitary on.
\end{proof}

\paragraph{A single layer of random single-qubit Cliffords in negligibly distinct.}
While $\negl(n)$-distinctness is necessary for PRUs, it also suffices in the sense it suffices for $D \leftarrow \cal D$  to be $\negl(n)$-distinct in order for unitaries $U = PFD$ to form a PRU ensemble (\Cref{thm:negl-distinct-to-PRU}). 

\begin{proposition}[A single layer of single-qubit Cliffords is $(2/3)^n$-EAC]
\label{prop:single-layer-cliffords-eac}
Let $\cal E$ be the ensemble obtained by sampling
\begin{equation}
U=C_1\ot\cdots\ot C_n,
\end{equation}
where for each $i \in [n]$, $C_i \leftarrow \Cl_{\C}(2)$ are drawn independently from the single qubit complex Clifford group, or an exact single qubit $2$-design,
Then $\calE$ is $\delta$-EAC with $\delta=(2/3)^n$.
\end{proposition}
\begin{proof}
To show $(2/3)^n$-EAC, it suffices that for any bi-partite state $\omega_{\X_i, \X_j}$,
\begin{equation}
    \tr[\Pi^{\eq}_{\X_i, \X_j} \E_U[U^{\ot 2}\omega_{\X_i, \X_j} U^{\ot 2, \dagger}]] \leq (2/3)^n,
\end{equation}
where each $\X_i$ and $\X_j$ is an $n$-qubit register. Identify each qubit within a given n-qubit register by $b \in [n]$ s.t.\ $\X_{i,b}$ denotes $b$-th qubit within register $\X_i$. Similarly, for $X_{j,b}$. This allows us to write the single qubit equality projector as
\begin{equation}
    \pi_{\X_{i,b},\X_{j,b}}^{\eq}
    :=
    \sum_{a,c\in\zo}
    \delta_{a,c}\,
    \ketbra{a}{a}_{\X_{i,b}}
    \ot
    \ketbra{c}{c}_{\X_{j,b}}.
\end{equation}
The primary technical insight of the proof is realizing that the global equality projector is a product of the local ones. That is,
\begin{equation}\label{eq:global-eq-proj-factorizes}
\Pi^{\eq}_{\X_i, \X_j} = \prod_{b=1}^n  \left(\pi_{\X_{i,b},\X_{j,b}}^{\eq}\right).
\end{equation}
To see this, let us first rewrite $\pi_{\X_{i,b},\X_{j,b}}^{\eq}$ as a $2n$-qubit operator
\begin{equation}
    \pi_{\X_{i,b},\X_{j,b}}^{\eq} = \sum_{x, y \in \zo^n} \delta_{x_b,y_b} \ketbra{x}{x}_{\X_i} \ot \ketbra{y}{y}_{\X_j}.
\end{equation}

Indeed, for every computational basis $\ket{x}_{\X_i}\ot \ket{y}_{\X_j}$, where $x, y \in \zo^n$,
\begin{equation}
    \pi_{\X_{i,b},\X_{j,b}}^{\eq} \left(\ket{x}_{\X_i}\ot \ket{y}_{\X_j}\right) = \delta_{x_b, y_b} \ket{x}_{\X_i} \ot \ket{y}_{\X_j}.
\end{equation}
Hence,
\begin{equation}
    \prod_{b=1}^n  \left(\pi_{\X_{i,b},\X_{j,b}}^{\eq}\right) \left(\ket{x}_{\X_i}\ot \ket{y}_{\X_j}\right) = \left(\prod_{b=1}^n \delta_{x_b, y_b}\right) \left(\ket{x}_{\X_i}\ot \ket{y}_{\X_j}\right) = \delta_{x,y} \left(\ket{x}_{\X_i}\ot \ket{y}_{\X_j}\right),
\end{equation}
which is exactly $\Pi^{\eq}_{\X_i, \X_j} \left(\ket{x}_{\X_i}\ot \ket{y}_{\X_j}\right)$. This proves \Cref{eq:global-eq-proj-factorizes}. By \Holder's inequality, $\tr[\omega_{\X_i, \X_j}] =1$, and cyclicity of trace, we get the inequality $ \tr[\Pi^{\eq}_{\X_i, \X_j} \E_U[U^{\ot 2}\omega_{\X_i, \X_j} U^{\ot 2, \dagger}]] \leq \norm{\E_U [U^{\ot 2, \dagger} \Pi^{\eq}_{\X_i, \X_j} U^{\ot 2}] }_{\infty}$, which we shall upper bound by $(2/3)^n$. By \Cref{eq:global-eq-proj-factorizes},

\begin{equation}\label{eq:global-eq-}
U^{\ot 2, \dagger} \Pi^{\eq}_{\X_i, \X_j} U^{\ot 2}
    = U^{\ot 2, \dagger} \prod_{b=1}^n  \left(\pi_{\X_{i,b},\X_{j,b}}^{\eq}\right) U^{\ot 2} 
    =\prod_{b=1}^n U^{\ot 2, \dagger} \left(\pi_{\X_{i,b},\X_{j,b}}^{\eq}\right) U^{\ot 2} =  \prod_{b=1}^n C_b^{\ot 2, \dagger} \pi_{\X_{i,b},\X_{j,b}}^{\eq}  C_b^{\ot 2},
\end{equation}    
where the last equality holds since for each $b\in [n]$, $\pi_{\X_{i,b},\X_{j,b}}^{\eq}$ acts only on $\X_{i,b},\X_{j,b}$ and $U = C_1 \ot \cdots \ot C_n$, where all $C_{b'}$ with $b'\neq b$ commute through $\pi_{\X_{i,b},\X_{j,b}}^{\eq}$ to yield
\begin{equation} 
U^{\ot 2, \dagger} \pi_{\X_{i,b},\X_{j,b}}^{\eq} U^{\ot 2} = C_b^{\ot 2, \dagger} \pi_{\X_{i,b},\X_{j,b}}^{\eq}  C_b^{\ot 2}.
\end{equation}
Since for each $b \in [n]$, $C_b \leftarrow \Cl_{\C}(2)$ is drawn independently from $1$-qubit complex Clifford group, 
\begin{equation}
    \E_U U^{\ot 2, \dagger} \Pi^{\eq}_{\X_i, \X_j} U^{\ot 2} =  \prod_{b=1}^n \E_{C_b \leftarrow \Cl_{\C}(2)} C_b^{\ot 2, \dagger} \pi_{\X_{i,b},\X_{j,b}}^{\eq}  C_b^{\ot 2}.
\end{equation}
By a standard (single-qubit) 2-design twirl calculation \cite{mele_introduction_2024}, $\E_{C_b \leftarrow \Cl_{\C}(N)} C_b^{\ot 2, \dagger} \pi_{\X_{i,b},\X_{j,b}}^{\eq}  C_b^{\ot 2} = \Id/3 + \SWAP/3$. Thus, $\norm{\E_{C_b \leftarrow \Cl_{\C}(N)} C_b^{\ot 2, \dagger} \pi_{\X_{i,b},\X_{j,b}}^{\eq}  C_b^{\ot 2}}_{\infty} \leq 2/3$. Finally submultiplicativity of the operator norm, 
\begin{equation}
\norm{\E_U U^{\ot 2, \dagger} \Pi^{\eq}_{\X_i, \X_j} U^{\ot 2}}_{\infty} \leq (2/3)^n, 
\end{equation}
as desired.
\end{proof}

\section{Necessary quantum resources for the distinctness}
\begin{theorem}[Imaginarity is necessary for distinctness]\label{thm:imag-necc-for-dist}
    Let $\calE$ be an ensemble of $n$-qubit unitaries. Define the imaginarity of unitary $U$ as $I_p(U) := 1-\frac{1}{N^2}\bigl|\tr[U^\dagger U^*]\bigr|^2$.

If $\calE$ is $\delta$-distinct, then
\begin{equation}
\E_{U\leftarrow\calE}[I_p(U)] \ge 1-\delta.
\label{eq:imaginarity-lb}
\end{equation}
\end{theorem}
\begin{proof}
By \Cref{lemma:distinctness-implies-EAC}, $\delta$-distinctness implies $\delta$-EAC.
Hence it suffices to prove the stated inequality under the assumption that
$\calE$ is $\delta$-EAC.
By \Cref{def:delta-EAC}, this means that for every bipartite state $\omega_{\X_i,\X_j}$ on two $n$-qubit
registers $\X_i$ and $\X_j$,
\begin{equation}
\tr\left[
\Pi^{\mathrm{eq}}_{\X_i,\X_j}\,
\E_{U\leftarrow\calE}
\Bigl[
U^{\ot 2}\,\omega_{\X_i,\X_j}\,U^{\dagger\ot 2}
\Bigr]
\right]
\le \delta.
\label{eq:eac-assumption}
\end{equation}
We now choose a special bipartite input state. Let
\begin{equation}
\ket{\Phi_N}_{\X_i,\X_j}
:=
\frac{1}{\sqrt{N}}
\sum_{a=0}^{N-1} \ket{a,a},
\qquad
\omega^{\mathrm{im}}_{\X_i,\X_j}
:=
\ket{\Phi_N}\bra{\Phi_N}_{\X_i,\X_j}.
\label{eq:max-ent-state}
\end{equation}
For a fixed unitary $U$, define
\begin{equation}
p_{\mathrm{im}}(U)
:=
\tr\left[
\Pi^{\mathrm{eq}}_{\X_i,\X_j}\,
U^{\ot 2}\,\omega^{\mathrm{im}}_{X_i,X_j}\,U^{\dagger\ot 2}
\right].
\end{equation}
Then
\begin{align}
p_{\mathrm{im}}(U)
=
\sum_{x=0}^{N-1}
\left|
\bra{x,x}U^{\ot 2}\ket{\Phi_N}
\right|^2
=
\frac{1}{N}
\sum_{x=0}^{N-1}
\left|
\sum_{a=0}^{N-1} U_{x ,a}U_{x ,a}
\right|^2
=
\frac{1}{N}
\sum_{x=0}^{N-1}
\bigl|(UU^T)_{x, x}\bigr|^2.
\end{align}
By the Cauchy-Schwarz inequality, we get
\begin{equation}
\sum_{x=0}^{N-1} |z_x|^2 \ge \frac{1}{N}\left|\sum_{x=0}^{N-1} z_x\right|^2
\end{equation}
with $z_x=(UU^T)_{x ,x}$ gives
\begin{equation}\label{eq:pim-imag-bound}
p_{\mathrm{im}}(U)
\ge
\frac{1}{N^2}
\left|
\sum_{x=0}^{N-1} (UU^{\mathsf{T}})_{x ,x}
\right|^2
=
\frac{1}{N^2}
\bigl|\tr[UU^{\mathsf{T}}]\bigr|^2
=\frac{1}{N^2}\bigl|\overline{\tr[UU^{\mathsf{T}}]}\bigr|^2
=
\frac{1}{N^2}
\bigl|\tr[U^\dagger U^*]\bigr|^2
=
1-I_p(U),
\end{equation}
where we have used that
$U^* U^\dagger  = (UU^{\mathsf{T}})^\dagger$, cyclicity of trace and that, for any operator $X$, we have that $\overline{\tr[X]} = \tr[X^{\dagger}]$. Now using $\omega^{\mathrm{im}}_{\X_i,\X_j}$ in the definition of $\delta$-EAC and substituting \Cref{eq:pim-imag-bound}, we get
\begin{align}
\delta
\ge
\tr\left[
\Pi^{\mathrm{eq}}_{\X_i,\X_j}\,
\E_{U\leftarrow\calE}
\Bigl[
U^{\ot 2}\,\omega^{\mathrm{im}}_{\X_i,\X_j}\,U^{\dagger\ot 2}
\Bigr]
\right]
=
\E_{U\leftarrow\calE}\left[p_{\mathrm{im}}(U)\right]
\ge
1-\E_{U\leftarrow\calE}[I_p(U)].
\end{align}
Rearranging yields the claimed inequality.
\end{proof}

\begin{theorem}[Coherence is necessary for distinctness]\label{thm:coh-necc-for-dist}
    Let $\calE$ be an ensemble of $n$-qubit unitaries. Define coherence of a unitary $U$ as $C_p(U) := -\frac{1}{N}\sum_{x,y=0}^{N-1} |U_{x, y}|^2 \ln |U_{x ,y}|^2$.
If $\calE$ is $\delta$-distinct. Then,
\begin{equation}
\E_{U\leftarrow\calE}[C_p(U)] \ge \ln(1/\delta).
\label{eq:coherence-lb}
\end{equation}
\end{theorem}
\begin{proof}
We will proceed in a similar way as we did in the proof of \Cref{thm:imag-necc-for-dist}. Again, by \Cref{lemma:distinctness-implies-EAC}, $\delta$-distinctness implies $\delta$-EAC.
Hence it suffices to prove the stated inequality under the assumption that
$\calE$ is $\delta$-EAC. Once again, by definition of $\delta$-EAC (\Cref{def:delta-EAC}),
\begin{equation}
\tr\left[
\Pi^{\mathrm{eq}}_{\X_i,\X_j}\,
\E_{U\leftarrow\calE}
\Bigl[
U^{\ot 2}\,\omega_{\X_i,\X_j}\,U^{\dagger\ot 2}
\Bigr]
\right]
\le \delta
\end{equation}
for every input state $\omega_{\X_i,\X_j}$. Similar to in the proof of \Cref{thm:imag-necc-for-dist}, we will choose a special bipartite input state. Let
\begin{equation}
\omega^{\mathsf{coh}}_{\X_i,\X_j}
:=
\frac{1}{N}
\sum_{y=0}^{N-1} \ket{y,y}\bra{y,y}_{\X_i,\X_j}.
\label{eq:coh-input}
\end{equation}
For a fixed unitary $U$, define
\begin{equation}
p_{\mathsf{coh}}(U)
:=
\tr\left[
\Pi^{\mathrm{eq}}_{\X_i,\X_j}\,
U^{\ot 2}\,\omega^{\mathsf{coh}}_{\X_i,\X_j}\,U^{\dagger\ot 2}
\right].
\label{eq:def-p-coh}
\end{equation}
A direct computation gives that
\begin{align}
p_{\mathsf{coh}}(U)
&=
\frac{1}{N}
\sum_{y=0}^{N-1}
\sum_{x=0}^{N-1}
|U_{x ,y}|^4.
\label{eq:pcoh-1}
\end{align}
Now, for each fixed column $y$, define the probability distribution $p_x^{(y)} := |U_{x ,y}|^2$
s.t. $\sum_{x=0}^{N-1} p_x^{(y)} = 1$.
Then \Cref{eq:pcoh-1} can be written as 
\begin{equation}
p_{\mathsf{coh}}(U) = \frac{1}{N}\sum_{y=0}^{N-1} Q_y,
\label{eq:pcoh-Q}
\end{equation}
where $Q_y := \sum_{x=0}^{N-1} \bigl(p_x^{(y)}\bigr)^2$ is the collision probability. Denoting the Shannon entropy of the distribution as $H_y := -\sum_{x=0}^{N-1} p_x^{(y)} \ln p_x^{(y)}$, we claim that
\begin{equation}
Q_y \ge e^{-H_y}
\qquad\text{for every } y.
\label{eq:Qy-lower}
\end{equation}
Indeed, by concavity of $\ln$ and the fact that $\sum_x p_x^{(y)}=1$,
\begin{align}
\sum_{x=0}^{N-1} p_x^{(y)} \ln p_x^{(y)}
&\le
\ln\left(
\sum_{x=0}^{N-1} p_x^{(y)} p_x^{(y)}
\right)
=
\ln Q_y.
\label{eq:log-concavity-step}
\end{align}
Multiplying by $-1$ and exponentiating yields \Cref{eq:Qy-lower}. Substituting \Cref{eq:Qy-lower} into \Cref{eq:pcoh-Q}, we obtain
\begin{align}\label{eq:pcoh-lower-bound}
p_{\mathsf{coh}}(U)
\ge
\frac{1}{N}\sum_{y=0}^{N-1} e^{-H_y}
\ge
\exp\left(
-\frac{1}{N}\sum_{y=0}^{N-1} H_y
\right)
=
e^{-C_p(U)},
\end{align}
where the second inequality is Jensen's inequality, since $e^{-x}$ is convex, and
\begin{equation}
\frac{1}{N}\sum_{y=0}^{N-1} H_y
=
-\frac{1}{N}\sum_{x,y=0}^{N-1}|U_{x , y}|^2 \ln |U_{x , y}|^2
=
C_p(U).
\end{equation}
Now using $\omega^{\mathsf{coh}}_{\X_i,\X_j}$ in the definition of $\delta$-EAC (\Cref{def:delta-EAC}) and substituting \Cref{eq:pcoh-lower-bound}, we get
\begin{align}\label{eq:coh-exp-lb}
\delta
\ge
\tr\left[
\Pi^{\mathrm{eq}}_{\X_i,\X_j}\,
\E_{U\leftarrow\calE}
\Bigl[
U^{\ot 2}\,\omega^{\mathsf{coh}}_{\X_i,\X_j}\,U^{\dagger\ot 2}
\Bigr]
\right]
=
\E_{U\leftarrow\calE}\left[p_{\mathsf{coh}}(U)\right]
\ge
\E_{U\leftarrow\calE}\left[e^{-C_p(U)}\right].
\end{align}
Again, due to convexity of $e^{-x}$, Jensen's inequality gives
\begin{equation}
\E_{U\leftarrow\calE}\left[e^{-C_p(U)}\right]
\ge
e^{-\E_{U\leftarrow\calE}[C_p(U)]}.
\label{eq:jensen-final}
\end{equation}
Combining \Cref{eq:coh-exp-lb} and \Cref{eq:jensen-final}, we have $\delta \ge e^{-\E_{U\leftarrow\calE}[C_p(U)]}$. Hence
\begin{equation*}
    \E_{U\leftarrow\calE}[C_p(U)] \ge \ln\Bigl(\frac{1}{\delta}\Bigr).
\end{equation*}
\end{proof}

\Cref{thm:imag-necc-for-dist} and \Cref{thm:coh-necc-for-dist} allow us 
to recover the constraints on coherence and imaginarity for PRUs worked out in Ref.\  \cite{haug_pseudorandom_2024}. 

\begin{corollary}[PRUs must be coherent and imaginary {\cite[Theorems 3 and 5]{haug_pseudorandom_2024}}]\label{corr:PRUs-negl-imag-and-coh-from-HBK[24]}
    If $\cal E$ is an n-qubit PRU ensemble, then
\begin{equation}
\E_{U\leftarrow\calE}[I_p(U)] = 1-\mathrm{negl}(n)
\qquad\text{and}\qquad
\E_{U\leftarrow\calE}[C_p(U)] = \omega(\log n).
\label{eq:negl-consequence}
\end{equation}
\end{corollary}
\begin{proof}
   \Cref{thm:imag-necc-for-dist} and \Cref{thm:coh-necc-for-dist} establish that if an ensemble $\cal E$ is $\delta$-distinct then $\E_{U\leftarrow\calE}[I_p(U)] \ge 1-\delta$ and $\E_{U\leftarrow\calE}[C_p(U)] \ge \ln\Bigl(\frac{1}{\delta}\Bigr)$, respectively. Additionally, from \Cref{thm:PRU-must-be-negl-dist}, we know that any PRU ensemble $\cal E$ must be $\negl(n)$-distinct. Thus, setting $\delta=\negl(n)$ yields the claim.
\end{proof}

\section{Real valued PRUs beyond PPT input states}
\label{sec:real-prus}
Existence of real-valued PRUs has been ruled out in 
Ref.\ \cite{haug_pseudorandom_2024}. Later, Brakerski and Magrafta \cite{brakerski_real-valued_2024} proposed an explicit ensemble of real-valued PRUs (composed of random real binary phase, Hadamard, and a random permutation) and showed that they are statistically indistinguishable from Haar random unitaries, so long as the adversary can only query on states any polynomial set of orthogonal input states or on states with 
high min-entropy in the computational basis. The distinguisher in Ref.\ \cite{haug_pseudorandom_2024} works by querying the unknown unitary on the maximally entangled state. It was, thus, proposed as an open question in Ref.\ \cite{brakerski_real-valued_2024} to construct real-valued PRUs for product states. We show that the $PF_{\R}C_{\R}$ ensemble with a real Clifford is indeed one such ensemble and $PF_{\R}HF_{\R}$ is another such ensemble. Recall that $F_{R} = \sum_{x\in[N]}(-1)^{f(x)}\ketbra{x}{x}$ is a binary phase operator, where $f$ is a uniformly random Boolean function. Note that $PF_{\R}HF_{\R}$ is a real-valued version of $PF_{\R}HF_{\C}$, which  is secure on arbitrary input states (\Cref{thm:PFHF-PRU}). This is akin to how $PF_{\R}C_{\R}$ is an input restricted real-valued analogue of $PFC$, where $C$ in the latter case is a complex-valued unitary $2$-design.
We show that $PF_{\R}C_{\R}$ and $PF_{\R}HF_{\R}$  form a PRU and $\calO(t^2/N)$-approximate unitary $t$-design on a larger class of states (than just product or even positive partial transpose states) that we classify as those having a small Bell overlap, a measure we introduce in \Cref{def:bell-overlap}. To quote some examples, separable states across $t$ registers, in fact every PPT (\emph{positive partial transpose}) state, Choi states of traceless unitaries, which are NPT (\emph{non-positive partial transpose}) states,  have a 
constant (in $n$) Bell overlap, due to which we get a larger class of states on which $PFC_{\R}$ and $PF_{\R}HF_{\R}$ are pseudorandom. 
In order to avoid repetition between $PFC_{\R}$  and $PF_{\R}HF_{\R}$ ensembles, for the rest of this section, we will work mostly with $PFC_{\R}$ as it may be more familiar to the reader in the context of the original $PFC$ ensemble \cite{metger_simple_2024}, but our results hold identically for the $PF_{\R}HF_{\R}$ ensemble as well. 
Let us start by defining the Bell overlap condition.
\begin{definition}[Bell Overlap (BO)]\label{def:bell-overlap}

Let $\cal S$ be an ensemble of $nt$-qubit quantum states on $t$ many $n$-qubit registers $\X_1, \ldots, \X_t$. For $\rho_{\mathsf{X}_1, \ldots , \mathsf{X}_t} \in \cal S$ and distinct $i, j \in [t]$, denote its reduced state on registers $\X_i$ and $\X_j$ by 
\begin{equation}
\rho_{\X_i, \X_j} \coloneqq \tr_{\X_{[t] \backslash \{i, j\}}}[\rho_{\mathsf{X}_1, \ldots , \mathsf{X}_t}].
\end{equation} 
We say that $\cal S$ has maximum $f(n)$ Bell overlap, denoted BO$(f(n))$, if

\begin{equation}
   \beta(\calS) := \sup_{\substack{\rho_{\mathsf{X}_1, \ldots , \mathsf{X}_t} \in \cal S\\ 1\leq i < j\leq t}}\tr[\rho_{\X_i, \X_j}\ket{\Omega}\bra{\Omega}_{\X_i, \X_j}] \leq f(n).
\end{equation}
\end{definition}
Note that Bell overlap is not intended to limit entanglement (See  \Cref{remark:compare-grevink} for an example). It limits alignment with {\em one} additional operator that distinguishes the orthogonal and unitary second moments, i.e., the Bell state. Let us now recall the explicit form of the $t$-wise twirl w.r.t.\ Haar random orthogonal matrix.

\begin{lemma}[Second order twirl w.r.t.\ the orthogonal group \cite{garcia-martin_quantum_2025, collins_integration_2006,grevink_will_2025}]\label{lemma:orth2-twirl}
For any operator $X \in \calL((\C^N)^{\ot 2})$, and a Haar random orthogonal matrix $O_{\Haar}$,
\begin{align}
    \E_{O_{\Haar}}(O\ot O)X(O^{\dag} \ot O^{\dag}) &= c_{\mathbbm{1}} \mathbbm{1}\ot\mathbbm{1} + c_{\SWAP} \SWAP + c_{\Omega} \ketbra{\Omega}{\Omega},
\end{align}
where 
\begin{align}
    c_{\mathbbm{1}} = \frac{1}{N(N+2)(N-1)}((N+1)\tr[X] - \tr[X\cdot \SWAP] - \tr[X \cdot \ketbra{\Omega}{\Omega}]) ,\\
    c_{\SWAP} = \frac{1}{N(N+2)(N-1)}(-\tr[X] + (N+1)\tr[X\cdot \SWAP] - \tr[X \cdot \ketbra{\Omega}{\Omega}]) ,\\
    c_{\Omega} = \frac{1}{N(N+2)(N-1)}(-\tr[X] - \tr[X\cdot \SWAP] + (N+1)\tr[X \cdot \ketbra{\Omega}{\Omega}]).
\end{align}
\end{lemma}
We are ready to show that any orthogonal 2-design ensemble is $\mathcal{O}(t^2/N)$-distinct on state ensembles with constant Bell overlap.

\begin{lemma}[Distinctness of orthogonal $2$-designs at bounded Bell overlap]\label{lemma:orth-2design-dist}
    Let $\calE$ be an orthogonal $2$-design ensemble $\mathcal{E}$ and $\calS$ be an ensemble of states with Bell overlap $\mathcal{O}(1) $ (\Cref{def:bell-overlap}). Then $\calE$ is $\mathcal{O}(t^2/N)$-distinct w.r.t.\ $\calS$.
\end{lemma}

\begin{proof}
 It suffices to show that for any $\rho_{\mathsf{X}_1, \ldots , \mathsf{X}_t} \in \calS$, 
  \begin{equation}
        \tr[\Pi^{\dist}\E_{U \leftarrow \mathcal{E}}C^{\otimes t}\rho_{\mathsf{X}_1, \ldots , \mathsf{X}_j} C^{\dagger, \otimes t}] \geq 1 - \mathcal{O}\left( \frac{t^2}{N} \right).
\end{equation}
 Our strategy is to consider (and upper bound with respect to) the orthogonal complement of the distinct subspace projector $\overline{\Pi}^{\dist} \coloneqq \mathbbm{1} - {\Pi}^{\dist}
 $.  
 Starting from Eq.~(3.8) in Ref.~\cite{metger_simple_2024}, we have that 
    \begin{align}
        \tr[\overline{\Pi}^{\dist} \E_{C \leftarrow \mathcal{E}}C^{\otimes t}\rho_{\mathsf{X}_1, \ldots , \mathsf{X}_t} C^{\dagger, \otimes t}] &\leq \sum_{1\leq i<j \leq t} \tr[\sum_{x\in [N]}\ketbra{x}{x}^{\ot 2}\E_{C \leftarrow \cal E} (C\ot C)\rho_{\X_i, \X_j}(C^{\dag} \ot C^{\dag})]\\
        \nonumber
        &= \frac{t(t-1)}{2} \cdot \tr[\sum_{x\in [N]}\ketbra{x}{x}^{\ot 2}\E_{C \leftarrow \cal E} (C\ot C)\rho_{\X_i, \X_j}(C^{\dag} \ot C^{\dag})].
    \end{align}
    Unlike \cite[Lemma 3.2]{metger_simple_2024}, however, we must prove an {\em input state dependent} bound. For that, we will bound the right hand side by $\mathcal{O}(t^2/N)$. To that end, it suffices to bound 
    \begin{equation}
         \tr[{\Pi}^{\mathsf{eq}}\E_{C \leftarrow \cal E} (C\ot C)\rho_{\X_i, \X_j}(C^{\dag} \ot C^{\dag})] \leq \mathcal{O}(1/N),
    \end{equation}
    where ${\Pi}^{\mathsf{eq}} = \sum_{x\in [N]} \ketbra{x}{x} \otimes \ketbra{x}{x}$. Due to \Cref{lemma:orth2-twirl}, we can write 
\begin{align}
    \tr[{\Pi}^{\mathsf{eq}}\E_{C \leftarrow \cal E} (C\ot C)\rho_{\X_i, \X_j}(C^{\dag} \ot C^{\dag})] &= c_{\mathbbm{1}} \tr[{\Pi}^{\mathsf{eq}}] + c_{\SWAP} \tr[{\Pi}^{\mathsf{eq}} \SWAP] + c_{\Omega}\tr[{\Pi}^{\mathsf{eq}}\ketbra{\Omega}{\Omega}].
\end{align}
It is easy to verify that
\begin{equation}
    \tr[{\Pi}^{\mathsf{eq}}] = \tr[{\Pi}^{\mathsf{eq}} \SWAP] = \tr[{\Pi}^{\mathsf{eq}}\ketbra{\Omega}{\Omega}] = N.
\end{equation}
Thus,
\begin{equation}
     \tr[{\Pi}^{\mathsf{eq}}\E_{C \leftarrow \cal E} (C\ot C)\rho_{\X_i, \X_j}(C^{\dag} \ot C^{\dag})] = N ( c_{\mathbbm{1}} +  c_{\SWAP} + c_{\Omega}).
\end{equation}
Again, by \Cref{lemma:orth2-twirl}, we have that
\begin{align}
 \nonumber
    c_{\mathbbm{1}} +  c_{\SWAP} + c_{\Omega} &= \frac{\tr[\rho_{\X_i, \X_j}] + \tr[\rho_{\X_i, \X_j}\SWAP] + \tr[\rho_{\X_i, \X_j}\ketbra{\Omega}{\Omega}]}{N(N+2)}\\
    \nonumber
    &\leq \frac{2 + \tr[\rho_{\X_i, \X_j}\ketbra{\Omega}{\Omega}]}{N(N+2)}\\
    &\leq \frac{2 + \mathcal{O}(1)}{N(N+2)}\label{eq:BO-condition},
\end{align}
where the last inequality holds because the state ensemble $\calS$ is BO($\calO(1)$) by assumption.
Thus,

\begin{equation}
     \tr[{\Pi}^{\mathsf{eq}}\E_{C \leftarrow \cal E} (C\ot C)\rho_{\X_i, \X_j}(C^{\dag} \ot C^{\dag})] = \frac{2 + \mathcal{O}(1)}{N+2} = \mathcal{O}\left(\frac{1}{N}\right).
\end{equation}
Putting together, this gives
\begin{equation}
\tr[\overline{\Pi}^{\dist} \E_{C \leftarrow \mathcal{E}}C^{\otimes t}\rho_{\mathsf{X}_1, \ldots , \mathsf{X}_t} C^{\dagger, \otimes t}] \leq \mathcal{O}\left( \frac{t^2}{N} \right)
\end{equation} 
as desired. 
\end{proof}

Recall that the real Clifford group forms an orthogonal 2-design \cite{hashagen_real_2018}. \Cref{lemma:orth-2design-dist}, combined with \Cref{thm:negl-distinct-to-PRU} gives the following.

\begin{corollary}[$PFC_{\R}$ is a real unitary design for bounded Bell overlap inputs]\label{corr:PFC-real-PRU}
  Let $C_{R}\leftarrow \Cl_{\R}(N)$ be a real Clifford drawn uniformly randomly from the $n$-qubit real Clifford group $\Cl_{\R}(N)$ and $\calS$ be the set of quantum states with Bell overlap at most $\calO(1)$. Then for every $(\rho_{\X_1,\ldots,\X_t}) \in \calS$,
  \begin{equation}
    \left\|\calM_{PFC_{\R}}^{(t)}(\rho_{\X_1,\ldots,\X_t})-\calM_{U_{\Haar}}^{(t)}(\rho_{\X_1,\ldots,\X_t})\right\|_1 \leq \mathcal{O}(t/\sqrt{N}).
  \end{equation}
\end{corollary}

 We show in \Cref{cor:binary-HF-EAC} that $HF_{\R}$ ensemble is $\mathcal{O}(t^2/N)$-distinct on input states with $\calO(1)$ Bell overlap, without even being a state $1$-design. Hence, we similarly get that

\begin{corollary}[$PF_{\R}HF_{\R}$ is a real unitary design for bounded Bell overlap inputs]\label{corr:PFHF-real-PRU}
  Let 
  \begin{equation}
  F_{\R} \coloneqq \sum_{x\in[N]}(-1)^{f(x)}\ketbra{x}{x},
\end{equation}
where $f:\zo^n\to\zo$ be a uniformly random Boolean function and let $\calS$ be the set of quantum states with Bell overlap at most $\calO(1)$. Then for every $(\rho_{\X_1,\ldots,\X_t}) \in \calS$,
  \begin{equation}
    \left\|\calM_{PF_{\R}HF_{\R}}^{(t)}(\rho_{\X_1,\ldots,\X_t})-\calM_{U_{\Haar}}^{(t)}(\rho_{\X_1,\ldots,\X_t})\right\|_1 \leq \mathcal{O}(t/\sqrt{N}).
  \end{equation}
\end{corollary}

We can go even further and use the $PFC_{\R}$ (or $PF_{\R}HF_{\R}$) ensemble to show the equivalence of unitary and orthogonal twirls on state ensembles with constant Bell overlap. 

\begin{proposition}[Unitary–orthogonal twirl equivalence at bounded Bell overlap]\label{prop:equiv-orthog-unitary-twirls-BO}
    Let $\calS$ be an ensemble of states with Bell overlap $\calO(1)$. Then for every $\rho_{\X_1,\ldots,\X_t} \in \calS$,
    \begin{equation}
         \left\|\calM_{O_{\Haar}}^{(t)}(\rho_{\X_1,\ldots,\X_t})-\calM_{U_{\Haar}}^{(t)}(\rho_{\X_1,\ldots,\X_t})\right\|_1 \leq \mathcal{O}(t/\sqrt{N}).
    \end{equation}
\end{proposition}
\begin{proof}
Noting each $PFC_{\R}$ lies in $O(N)$ and by the invariance of the Haar measure, we get the following two equalities: $\calM_{O_{\Haar}}^{(t)} \circ \calM_{PFC_{\R}}^{(t)} = \calM_{O_{\Haar}}^{(t)}$ and  $\calM_{O_{\Haar}}^{(t)}\circ \calM_{U_{\Haar}}^{(t)} = \calM_{U_{\Haar}}^{(t)}$. Thus,
\begin{align}
     \left\|\calM_{O_{\Haar}}^{(t)}(\rho_{\X_1,\ldots,\X_t})-\calM_{U_{\Haar}}^{(t)}(\rho_{\X_1,\ldots,\X_t})\right\|_1 &=  \left\|\calM_{O_{\Haar}}^{(t)}\left(\calM_{PFC_{\R}}^{(t)}(\rho_{\X_1,\ldots,\X_t})-\calM_{U_{\Haar}}^{(t)}(\rho_{\X_1,\ldots,\X_t})\right)\right\|_1 \\
     &\leq   \left\|\calM_{PFC_{\R}}^{(t)}(\rho_{\X_1,\ldots,\X_t})-\calM_{U_{\Haar}}^{(t)}(\rho_{\X_1,\ldots,\X_t})\right\|_1,
     \nonumber
\end{align}
where the second inequality is due to contractivity of trace under CPTP maps. Finally, substituting the bound in \Cref{corr:PFC-real-PRU} yields the desired claim.
\end{proof}

\begin{remark}[Beyond the PPT barrier]\label{remark:compare-grevink}
\Cref{prop:equiv-orthog-unitary-twirls-BO} improves the best known (to 
the authors' knowledge) bound concerning the equivalence of unitary and orthogonal twirls derived as \cite[Theorem S18]{grevink_will_2025}. \Cref{prop:equiv-orthog-unitary-twirls-BO}  holds for a larger class of input states, unlike only the PPT (\emph{positive partial transpose}) 
states in Ref.\ \cite[Theorem S18]{grevink_will_2025}, thus going beyond the PPT barrier identified as one of the open questions in Ref.\ \cite{grevink_will_2025}. The set of PPT states is a strictly smaller set than the set of states with Bell overlap $\calO(1)$. To see this, first note that \cite[Eq.~(S51)]{grevink_will_2025} implies that every PPT state, where the partial transpose is taken on each of the $t$ registers, has $\calO(1)$ Bell overlap. Conversely, there is a class of states with $\calO(1)$ Bell overlap that are NPT (\emph{non-positive partial transpose}). Concretely, let $W$ to be any traceless unitary and consider its Choi state vector 
\begin{equation}
\ket{\psi_W} \coloneqq (W \ot \Id)\ket{\Omega}/\sqrt{N}, 
\end{equation}
where
\begin{equation}
\ketbra{\psi_W}{\psi_W}^{\mathsf{T}_2} = \frac{1}{N}(W\ot \Id)\SWAP(W^{\dagger}\ot \Id).
\end{equation}
For any anti-symmetric state 
vector $\ket{\psi_{-}}$,
we have
\begin{equation}
\ketbra{\psi_W}{\psi_W}^{\mathsf{T}_2}((W\ot \Id)\ket{\psi_{-}}) = -\frac{1}{N}(W\ot \Id)\ket{\psi_{-}}. 
\end{equation}
Hence, all states of the form $\ketbra{\psi_W}{\psi_W}$ are NPT. Notably, by the tracelessness of the unitary $W$, all such states have Bell overlap exactly $0$. This is notable since small Bell overlap can be construed to imply low entanglement. This is not true. The states $\ketbra{\psi_W}{\psi_W}$ are maximally entangled and yet have $0$ Bell overlap. This is because the Bell overlap does not measure overlap w.r.t.\ any (maximally) entangled state, but only with the Bell state, motivated by the fact that it appears as an extra element in the second order commutant of the orthogonal group, compared to the unitary group.     
\end{remark}

While \Cref{remark:compare-grevink} identifies concrete states, e.g.,  PPT states, some NPT (Choi) states, with constant Bell overlap on which $PFC_{\R}$ (or $PF_{\R}HF_{\R}$) forms an orthogonal design, it is natural to ask what other states have constant Bell overlap and whether we can go beyond the constant regime. In the following, we shall address both these questions. Note that if we are not particular about the $t$-dependence in the trace norm bounds in \Cref{prop:equiv-orthog-unitary-twirls-BO} and \Cref{corr:PFC-real-PRU}, as is the case with PRUs, we can, for instance, accommodate states with Bell overlap $\poly(n)$ (see \cref{eq:BO-condition} in \Cref{lemma:orth-2design-dist}). The next 
natural question is to characterize such states. 
We do so by using the Schmidt number of bipartite mixed states as proposed by Terhal and Horodecki \cite{terhal_schmidt_2000}, as an extension of the Schmidt rank for bipartite pure states. 

\begin{definition}[Schmidt number (Definition 1 in Ref.\ \cite{terhal_schmidt_2000})]
    A bipartite density matrix $\rho$ has Schmidt number $k$ if $(i)$ for any decomposition of $\rho$, $\{p_i\geq 0, \ket{\psi_i}\}$ with $\rho = \sum_i p_i \ketbra{\psi_i}{\psi_i}$ at least one of the vectors has at least Schmidt rank $k$ and $(ii)$ there exists a decomposition of $\rho$ with all vectors $\{\ket{\psi_i}\}$ of Schmidt rank at most $k$.
\end{definition}

As one would expect, for a pure state, the Schmidt number reduces to its Schmidt rank. Since we work with $t$-partite ($n$-qubit) states, in the following we first define the appropriate $t$-copy extension for our purposes. 

\begin{definition}[Pairwise $t$-partite Schmidt-number class]
Let $S_k$ denote the set of bipartite density matrices with Schmidt
number at most $k$. We define its pairwise $t$-partite extension by
\begin{equation}
S_{k}^{(t)}
:=
\left\{
\rho_{\X_1,\ldots,\X_t}
:
\rho_{\X_i,\X_j}\in S_k
\text{ for every }1\le i<j\le t
\right\},
\end{equation}
where
\begin{equation}
\rho_{\X_i, \X_j} \coloneqq \tr_{\X_{[t] \backslash \{i, j\}}}[\rho_{\mathsf{X}_1, \ldots , \mathsf{X}_t}].
\end{equation}
\end{definition}

\begin{theorem}[Schmidt number bounds Bell overlap]\label{thm:SN-imply-BO}
    Any state $\rho_{\X_1,\ldots,\X_t} \in S_k^{(t)}$ has Bell overlap at most $k$.
\end{theorem}
\begin{proof}
    For any bipartite state $\rho_{\X_i,\X_j}$, denote its Schmidt number by  $\operatorname{SN}
(\rho_{\X_i,\X_j})$. Then, $S_{k}^{(t)}$ can equivalently be written as
\begin{equation}
S_{k}^{(t)}
=
\left\{
\rho_{\X_1,\cdots,\X_t}:
\operatorname{SN}
(\rho_{\X_i,\X_j})\le k
\text{ for every }1\le i<j\le t
\right\}.
\end{equation}
Due to \cite[Lemma 1]{terhal_schmidt_2000}, we have that $\operatorname{SN}
(\rho_{\X_i\X_j})\le k$ implies $\tr[\ketbra{\Omega}{\Omega}\rho_{\X_i,\X_j}] \leq k$. Since this holds for every bipartition $1\leq i<j\leq t$, the desired claim follows.
\end{proof}

\begin{theorem}[Real unitary designs on separable inputs]\label{thm:PFC-sep-states}
    The ensembles $PFC_{\R}$ and $PF_{\R}HF_{\R}$ are both $\mathcal{O}(t/\sqrt{N})$-approximate additive unitary $t$-design on all separable states across registers $\X_1,\ldots,\X_t$.
\end{theorem}
\begin{proof}
    We will show that the 
    set of separable states across $\X_1,\ldots,\X_t$ has Bell overlap at most $1$ and the claim will follow from \Cref{corr:PFC-real-PRU} and \Cref{corr:PFHF-real-PRU}. Since every $t$-copy separable state is also separable on any bipartition, and noting that the Schmidt number of a separable state is 1, by \Cref{thm:SN-imply-BO} we get that the set of separable states across $\X_1,\ldots,\X_t$ also has Bell overlap at most $1$.
\end{proof}

As a special case, \Cref{thm:PFC-sep-states} resolves the conjecture by \cite{brakerski_real-valued_2024} that posited the existence of real-valued PRUs on product input states. 
Note that for PRUs, we can tolerate input states with Bell overlap of up to $N/n^{\omega(1)}$ since a $\negl(n)$ trace distance upper bound suffices (\Cref{thm:negl-distinct-to-PRU}). Consequently, for PRUs, we have security against input states whose Schmidt number can be (at most) $N/n^{\omega(1)}$.

\section{Distinctness without designs}
 In this section, we will show that $HF$ ensemble is $\calO(1/N)$-EAC. This would directly imply $\calO(t^2/N)$-distinctness by \Cref{lemma:EAC-implies-distinctness} that has  earlier been achieved 
 only by a unitary 2-design. We stress again $HF$ is not even a state 1-design. We first show that $HF$ ensemble is $\calO(1/N)$-EAC for complex $F$ in \Cref{lemma:HF-distinct} and then proceed to an analogous statement on restricted input states in the case when $F$ is a real binary phase operator \Cref{cor:binary-HF-EAC}. 

\begin{lemma}\label{lemma:HF-distinct}
    Let $N\coloneqq 2^{n}, H\coloneqq H^{\ot n}$ be the $n$-qubit Hadamard, and  $F = \sum_{x \in [N]} \omega^{f(x)}\ketbra{x}{x}$ where $\omega = e^{2\pi i/3}$ and $f: \zo^n \rightarrow \{0, 1, 2\}$ is random ternary function. Then, $HF$ is $\calO(1/N)$-EAC.
\end{lemma}

\begin{proof}
    To show that an ensemble is $\calO(1/N)$-EAC, it suffices to show
    \begin{align}
        \tr[{\Pi}^{\mathsf{eq}}\E_{U\leftarrow \cal E}U^{\ot 2}\rho_{\X_i, \X_j}U^{\ot 2, \dagger}] \leq \mathcal{O}(1/N),
    \end{align}
    where recall that ${\Pi}^{\mathsf{eq}} = \sum_{x\in [N]} \ketbra{x}{x} \otimes \ketbra{x}{x}$ is the `equality projector' and $\rho_{\X_i, \X_j}$ is any bipartite quantum state. By Hölder's inequality, linearity of expectation, and normalization of $\rho_{\X_i, \X_j}$ it suffices to show 
    \begin{equation}
         \left\lVert\E_{U\leftarrow \cal E}[U^{\ot 2, \dagger}{\Pi}^{\mathsf{eq}}U^{\ot 2}]\right\rVert_{\infty} \leq \mathcal{O}(1/N).
    \end{equation}
    For the purposes of this analysis, we will show the above operator norm bound for $U=HFH$ and later get the same bound for $HF$ due to the unitary invariance of the operator norm, using the steps
    \begin{align}
        U_{a,b} &= \sum_{x,y \in [N]} H_{a,x}F_{x,y}H_{y,b} \\
        \nonumber
        &= \sum_{x \in [N]} H_{a,x}\omega^{f(x)}H_{x,b} \\
         \nonumber
        &= \sum_{x \in [N]} \frac{1}{\sqrt{N}}(-1)^{a\cdot x}\omega^{f(x)} \frac{1}{\sqrt{N}}(-1)^{b \cdot x} \\
         \nonumber
        &= \frac{1}{N} \sum_{x\in [N]} \omega^{f(x)} (-1)^{(a\oplus b)\cdot x},
         \nonumber
    \end{align}
    where $a.b = a_1b_1 + a_2b_2 + \ldots a_nb_n \mod 2$ and $a\oplus b$ is the bitwise XOR for two n-bit strings. 
    To avoid clutter, let us write $X(U) = \sum_{z\in [N]} U^{\dagger}\ketbra{z}{z}U \ot U^{\dagger}\ketbra{z}{z}U$. Then the elements of $X(U)$ can be written as 
    \begin{align}
        X(U)_{(a,b), (a',b')} &= \sum_{z \in [N]}\bra{a,b} U^{\dagger}\ketbra{z}{z}U \ot U^{\dagger}\ketbra{z}{z}U \ket{a',b'} \\
         \nonumber
        &= \sum_{z\in [N]}\bra{a}U^{\dagger}\ketbra{z}{z}U\ket{a'}\bra{b} U^{\dagger}\ketbra{z}{z}U\ket{b'}\\
         \nonumber
        &= \sum_{z\in [N]}\overline{U}_{z,a}U_{z,a'}\overline{U}_{z,b}U_{z,b'}
         \nonumber.
    \end{align}
    We can now substitute the general expression for the entries of $U$, to get
\begin{align}\label{eq:op-elements-HFH}
        X(HFH)_{(a,b), (a',b')} &= \sum_{z\in [N]} \overline{(N^{-1} \sum_{x\in [N]} \omega^{f(x)} (-1)^{(z\oplus a)\cdot x})}\cdot (N^{-1} \sum_{x'\in [N]} \omega^{f(x')} (-1)^{(z\oplus a')\cdot x'})\\
         \nonumber
        &\quad\cdot
         \overline{(N^{-1} \sum_{y\in [N]} \omega^{f(y)} (-1)^{(z\oplus b)\cdot y})}\cdot (N^{-1} \sum_{y'\in [N]} \omega^{f(y')} (-1)^{(z\oplus b')\cdot y'})\\
          \nonumber&=N^{-4}\sum_{\substack{z,x,y,x',y' \in [N]}} \overline{\omega^{f(x)}}\omega^{f(x')} \overline{\omega^{f(y)}} \omega^{f(y')} (-1)^{(z\oplus a)\cdot x} (-1)^{(z\oplus a')\cdot x'} (-1)^{(z\oplus b)\cdot y}(-1)^{(z\oplus b')\cdot y'}.
           \nonumber
    \end{align}
    Since $F$ is a random diagonal unitary, let us evaluate $\E_F[X(HFH)_{a,b), (a',b')}]$. For that, it is sufficient to analyse $\E_F[\overline{\omega^{f(x)}}\omega^{f(x')} \overline{\omega^{f(y)}} \omega^{f(y')}]$. First, observe that 
    \begin{equation}
        \E_F[\omega^{f(x)}] = \E_F[(\omega^{f(x)})^2] = 1/3\cdot(1+\omega+\omega^2) = 0,
    \end{equation}
    because $\omega^4 = \omega^3\cdot \omega = \omega$ and $(1+\omega+\omega^2) = 0$. Due to the vanishing first and second moments, it is easy to verify that 
\begin{equation}\label{eq:complex-F-twirl-omega}
        \E_F[\overline{\omega^{f(x)}}\omega^{f(x')} \overline{\omega^{f(y)}} \omega^{f(y')}] = \delta_{x,x'}\delta_{y,y'} + \delta_{x,y'}\delta_{y,x'} - \delta_{x,x'}\delta_{y,y'}\delta_{x,y},
    \end{equation}
where the last term subtracts the overcounting in the case $x=x'=y=y'$.
    Thus, 
    \begin{align}
        \E_{F}X(HFH)_{(a,b), (a',b')} &= N^{-4}\sum_{z,x,y \in [N]} (-1)^{(z\oplus a \oplus z\oplus a')\cdot x} (-1)^{(z\oplus b \oplus z\oplus b')\cdot y} + (-1)^{(z\oplus a \oplus z\oplus b')\cdot x} (-1)^{(z\oplus a' \oplus z\oplus b)\cdot y} \notag\\
         \nonumber
        &\quad - N^{-4}\sum_{z,x \in [N]}(-1)^{(z\oplus a \oplus z\oplus a' \oplus z\oplus b  \oplus z\oplus b')\cdot x}\\
         \nonumber
        &= N^{-3}\left(\sum_{x \in [N]} (-1)^{(a\oplus a')\cdot x}\sum_{y\in [N]}(-1)^{(b\oplus b')\cdot y} + \sum_{x \in [N]}(-1)^{(a\oplus b')\cdot x}\sum_{y\in [N]}(-1)^{(a'\oplus b)\cdot y}\right)  \notag\\
         \nonumber
        &\quad  - N^{-3} \sum_{x \in [N]}(-1)^{( a \oplus a' \oplus b \oplus b').x}\\
        &= N^{-3} \left((N \cdot \delta_{a,a'})(N \cdot \delta_{b,b'}) + (N \cdot \delta_{a,b'})(N \cdot \delta_{a',b}) -N \delta_{a \oplus a',b 
        \oplus b'} \right)  \\
        \nonumber
        &= N^{-1}(\delta_{a,a'}\delta_{b,b'} + \delta_{a,b'}\delta_{a',b}) - N^{-2}\delta_{a \oplus a',b \oplus b'},
         \nonumber
    \end{align}
    where the penultimate step follows from the orthogonality of characters of $\F_2^n$: $\sum_{x\in [N]}(-1)^{x.w} = N\delta_{w,0}$ (\Cref{fact:character-ortho}). Hence, 
    \begin{align}
        \E[X(HFH)] &= \sum_{a,a',b,b' \in [N]} \left(\frac{\delta_{a,a'}\delta_{b,b'} + \delta_{a,b'}\delta_{a',b}}{N} - \frac{\delta_{a \oplus a',b \oplus b'}}{N^2}\right) \ketbra{a,b}{a',b'} \\
         \nonumber
        &= \sum_{a,b \in [N]}
          \frac{\ketbra{a,b}{a,b} + \ketbra{a,b}{b,a}}{N} - \sum_{a,a',b,b' \in [N]}\delta_{a \oplus a',b \oplus b'}\cdot\frac{\ketbra{a,b}{a',b'}}{N^2} \\
           \nonumber
        &= \frac{1}{N}(I + \SWAP) - \frac{1}{N^2}J,
         \nonumber
    \end{align}
    where $J \coloneqq \sum_{a,a',b,b' \in [N]}\delta_{a \oplus a',b \oplus b'}\ketbra{a,b}{a',b'}$. We show in \Cref{lemma:J-is-PSD} that $J$ is PSD. Moreover, both $I + \SWAP$ and $ \E[X(HFH)]$ are also PSD operators. Thus, $\E[X(HFH)] \preceq \frac{1}{N}(I + \SWAP)$. Hence,
    \begin{equation}
        \left\lVert\E[X(HFH)]\right\rVert_{\infty} \leq \left\lVert\frac{1}{N}(I + \SWAP)\right\rVert_{\infty} = \frac{2}{N}.
    \end{equation}
    The claimed bound for $HF$ follows by the unitary invariance of operator norm
    w.r.t.\ the right-most Hadamard in $HFH$.
\end{proof}

We now show that if $F$ is a real-valued binary phase operator, then retains distinctness but on input states with bounded Bell overlap.

\begin{corollary}
\label{cor:binary-HF-EAC}
Let $F$ in \Cref{lemma:HF-distinct} be
\begin{equation}
F \coloneqq \sum_{x\in[N]}(-1)^{f(x)}\ketbra{x}{x},
\end{equation}
where $f:\zo^n\to\zo$ be a uniformly random Boolean function. Then, $HF$ is $\mathcal O(1/N)$-EAC on states with Bell overlap $\mathcal O(1)$.
\end{corollary}

\begin{proof}
Note that for real-valued binary phase operator  $F = \sum_{x\in[N]}(-1)^{f(x)}\ketbra{x}{x}$, \Cref{eq:complex-F-twirl-omega}, has two additional contributions
\begin{equation}
\delta_{x,y}\delta_{x',y'}
-
\delta_{x,x'}\delta_{y,y'}\delta_{x,y}.
\end{equation}
due to the non-vanishing second moment. Inserting these terms into \Cref{eq:op-elements-HFH} contributes
$N^{-1}\ketbra{\Omega}{\Omega}$ and $-N^{-2}J$, respectively. Hence,
the calculation in \Cref{lemma:HF-distinct} gives
\begin{equation}
\begin{aligned}
\E\left[X(HFH)\right]
&=
\frac{1}{N}
\left(
I+\SWAP+\ketbra{\Omega}{\Omega}
\right)
-\frac{2}{N^2}J \\
&\preceq
\frac{1}{N}
\left(
I+\SWAP+\ketbra{\Omega}{\Omega}
\right),
\end{aligned}
\end{equation}
where the inequality follows from the fact that $J$ is PSD (\Cref{lemma:J-is-PSD}). Therefore, for any state $\sigma_{\X_i, \X_j}$ with $\mathcal{O}(1)$ Bell overlap, we have that
\begin{equation}
\begin{aligned}
\tr\left[
\sigma_{\X_i, \X_j}\E\left[X(HFH)\right]
\right]
&\leq
\frac{
\tr[\sigma_{\X_i, \X_j}]
+\tr[\sigma_{\X_i, \X_j}\SWAP]
+\bra{\Omega}\sigma_{\X_i, \X_j}\ket{\Omega}
}{N} \label{eq:tr-HFH-real-bound}\\
&\leq
\frac{2+\mathcal{O}(1)}{N},
\end{aligned}
\end{equation}
where we used $\tr[\sigma_{\X_i, \X_j}]=1$ and $\tr[\sigma_{\X_i, \X_j}\SWAP]\leq 1$.
Finally, given $\rho_{\X_i,\X_j}$, define
\begin{equation}
\sigma_{\X_i, \X_j}
\coloneqq
H^{\ot 2}\rho_{\X_i,\X_j}H^{\ot 2}.
\end{equation}
Then,
\begin{equation}
(HF)^{\ot 2}
\rho_{\X_i,\X_j}
(HF^\dagger)^{\ot 2} \\
 =
(HFH)^{\ot 2}
\sigma_{\X_i, \X_j}
\bigl((HFH)^\dagger\bigr)^{\ot 2}.
\end{equation}
Thus, the bound in \Cref{eq:tr-HFH-real-bound} remains unchanged when using $HF$. Finally, to make sure that the Bell overlap assumption continues to hold, note that $H^{\ot 2}\ket{\Omega}=\ket{\Omega}$, and hence the Bell overlap stays unchanged:
$\bra{\Omega}\sigma\ket{\Omega}
=
\bra{\Omega}\rho_{\X_i,\X_j}\ket{\Omega}
\leq \mathcal{O}(1).
$
Applying the bound in \Cref{eq:tr-HFH-real-bound} to $\sigma_{\X_i, \X_j}$ proves the claim for $HF$.
\end{proof}

\begin{lemma}[Positivity of the XOR-collision operator]\label{lemma:J-is-PSD}
    $J \coloneqq \sum_{a,a',b,b' \in [N]}\delta_{a \oplus a',b \oplus b'}\ketbra{a,b}{a',b'}$ is \emph{positive semi-definite} (PSD).
\end{lemma}

\begin{proof}
   First, note that $\delta_{a \oplus a',b \oplus b'} = \delta_{a \oplus b,a' \oplus b'}$ because
\begin{equation}
a\oplus a' = b\oplus b'
\Longleftrightarrow
(a\oplus a')\oplus(a'\oplus b)=(b\oplus b')\oplus(a'\oplus b)
\Longleftrightarrow
a\oplus b = a'\oplus b',
\end{equation}
where we used associativity/commutativity of $\oplus$ and $x\oplus x=0$. Hence $J$ can be re-written as
\begin{equation}
    J = \sum_{a,a',b,b' \in [N]}\delta_{a \oplus b,a' \oplus b'}\ketbra{a,b}{a',b'}.
\end{equation}
Using the identity $\delta_{x, y} = \sum_{s\in [N]} \delta_{x,s}\delta_{y,s}$, rewrite
\begin{align}
J
&=\sum_{s \in [N]}\sum_{a,b \in [N]}\sum_{a',b' \in [N]} \delta_{a\oplus b,s}\,\delta_{a'\oplus b',s}\,|a,b\rangle\langle a',b'|
\label{eq:P-expand}\\
&=\sum_{s} |v_s\rangle\langle v_s|,
\label{eq:P-rankone}
\end{align}
where
\begin{equation}
|v_s\rangle := \sum_{a,b\in [N]:\,a\oplus b=s} |a, b\rangle .
\label{eq:def-vs}
\end{equation}
Thus,
\begin{equation}
J\succeq 0,
\label{eq:P-psd}
\end{equation}
since it is a sum of rank-one PSD operators.
\end{proof}

\section{Constraints on the JLS conjecture}\label{sec:alt-FH-distinguisher}
The distinctness of $HF$ informed us that a complex phase unitary followed a layer of Hadamard takes us to the distinct subspace. One might wonder if one could keep alternating similarly between independent phase-Hadamard layers to get a full PRU. In fact, Ref.\ 
\cite{ji_pseudorandom_2018} has 
conjectured 
that a constant number of alternating 
iterations of Hadamards and random phase unitaries result in a pseudorandom unitary. We show in 
the following that this is not true when the functions' co-domain has size  superpolynomially smaller 
than its domain size. 

\begin{theorem}[Small phase alphabets preclude pseudorandomness]\label{thm:two-copy-suffices}
Let $n, k\in\N$, $K\coloneqq 2^k \leq N\coloneqq 2^n$,
$\omega_N\coloneqq e^{2\pi i/N}$, and $l=\poly(n)$. For each layer $j\in[l]$, let
\begin{equation}
F_j=\sum_{x\in\{0,1\}^n}\omega_N^{\,f_j(x)}\ketbra{x}{x},\qquad f_j:[N]\to \{0,1,\ldots, K-1\},
\end{equation}
with each $f_j$ chosen uniformly randomly 
from the set of functions from $[N]\to \{0,1,\ldots, K-1\}$ where $\log K \leq n - \omega(\log n)$. If $\cal E$ is the ensemble of 
unitaries of 
the form 
\begin{equation}U=F_l H\cdots F_1 H,\end{equation}
then $\cal E$ cannot form a PRU ensemble.
\end{theorem}

\begin{proof}
A simple Bell-type distinguisher 
tells apart $\cal E$ from the Haar ensemble by first querying the unknown unitary on a Bell state and then measuring in the Bell basis. The success probability of any unitary $U$ for this test can be written as
\begin{equation}
p(U) = \braket{\Phi|(U\otimes U)\,\ketbra{\Phi}{\Phi}\,(U\otimes U)^{\dagger}|\Phi}
=\frac{1}{N^2}\,\big|\Tr[UU^{\mathsf{T}}]\big|^2,\qquad \ket{\Phi}=\frac{1}{\sqrt N}\sum_x\ket{x}\otimes\ket{x}.    
\end{equation}
It suffices to show that  \begin{equation}\lvert \E_{U\leftarrow U_{\Haar}}p(U) - \E_{U\leftarrow \cal E} p(U) \rvert \geq 1/\poly(n). \end{equation} For a Haar random unitary,
on the one hand, 
\begin{equation}\label{eq:p_succ-Haar}
    \E_{U\leftarrow U_{\Haar}}p(U) = \frac{2}{N(N+1)}\bra{\Phi}(\Id + \SWAP)\ket{\Phi} = \frac{2}{N(N+1)}.
\end{equation}   
On the other hand, due to Jensen's inequality, 
\begin{equation}\label{eq:p_succ-E}
    \E_{U\leftarrow \cal E} p(U) = \frac{1}{N^2}\E_{U \leftarrow \cal E}\big|\Tr[UU^{\mathsf{T}}]\big|^2 \geq \frac{1}{N^2}\big|\Tr[\E_{U \leftarrow \cal E}UU^{\mathsf{T}}]\big|^2.
\end{equation}
Since all $F_1, \ldots, F_l$ are i.i.d., $H^\mathsf{T} = H$, and $F_i^{\mathsf{T}}=F_i$, and so
\begin{equation}\label{eq:exp-UU^T}
    \E [UU^{\mathsf{T}}] = \E [(F_l H\cdots F_1 H)(H^\mathsf{T} F_1^{\mathsf{T}}\cdots H^{\mathsf{T}} F_l^{\mathsf{T}})] = (\E_{F_l}F_l^2)\ldots(\E_{F_1}F_1^2) = (\E_F [F^2])^l = \alpha^l\Id_N,
\end{equation}
where 
\begin{equation}
\alpha:= \E\Big[\omega_N^{\,2 f(x)}\Big]\ =\ \frac{1}{2^k}\cdot \sum_{a=0}^{2^k-1}\,e^{\frac{2\pi i}{N}\,2 a}.
\end{equation}
By the Dirichlet kernel identity (\Cref{fact:Dirichlet-kernel-identity}),
\begin{equation}\label{eq:alpha-dirichlet}
\left\lvert\alpha\right\rvert = \frac{1}{2^k}\cdot\left\lvert\frac{\sin\big(\pi\cdot \tfrac{2\cdot 2^{k}}{N}\big)}{\sin\big(\pi\cdot \tfrac{2}{N}\big)}\right\rvert = \frac{1}{2^k}\cdot\frac{\sin\big(\pi\cdot \tfrac{2\cdot 2^{k}}{N}\big)}{\sin\big(\pi\cdot \tfrac{2}{N}\big)},
\end{equation}
where for any large $n$ and $k = n - \omega(\log n)$ ensures that both the numerator and the denominator are positive and non-zero \footnote{Note that choosing $k=n-1$ will, due to $\sin(\pi)=0$, result in $\alpha =0$. Thus, vanishing the signal: $\alpha = 0$. Remarkably this step of the argument goes through with $k\leq n-2$. However, the ultimate choice of $k \leq n - \omega(\log n)$ is dictated by \Cref{eq:alpha^2-lower-bound} to account for the fact that $t$ could be an arbitrary polynomial in $n$.}. Same choice of $n$ and $k$ ensures that we can use the inequalities \begin{equation}x - x^3/6 < \sin(x) < x\end{equation} for $x \in (0, \pi/2]$ to get
\begin{equation}
    \abs{\alpha} > 1 - \frac{(L\theta)^2}{6},
\end{equation}
where $L= 2^k$ and $\theta = 2\pi/N$. Furthermore, 
\begin{equation}\label{eq:alpha^2-lower-bound}
    \abs{\alpha}^{2l} > \left(1 - \frac{(L\theta)^2}{6}\right)^{2l} \geq 1 - \frac{2l\cdot (L\theta)^2}{6} \geq 1 - \negl(n),
\end{equation}
where the penultimate inequality is a consequence of Bernoulli's inequality, given by
\begin{equation}(1+x)^r\geq1+rx\end{equation} for every $r\geq 1$ and any real $x\geq -1$ and the last inequality holds because, again, $k \leq n - \omega(\log n)$ and $l=\poly(n)$. Combining \Cref{eq:p_succ-E,eq:exp-UU^T}, we get that
\begin{equation}\label{eq:p_succ-E-final}
    \E_{U\leftarrow \cal E} p(U) \geq \abs{\alpha}^{2l} \geq 1 - \negl(n).
\end{equation} 
Finally, \Cref{eq:p_succ-E-final,eq:p_succ-Haar} imply 
\begin{equation}
    \left\lvert \E_{U\leftarrow U_{\Haar}}p(U) - \E_{U\leftarrow \cal E} p(U) \right\rvert \geq  \left\lvert 1 - \negl(n) - \frac{2}{2^n(2^n+1)} \right\rvert \geq 1 - \negl(n).
\end{equation}
\end{proof}

\newpage
\bibliographystyle{alpha}
\bibliography{refs_main,BigReferences70}
\end{document}